\documentclass{article}

\usepackage{arxiv}
\usepackage[utf8]{inputenc} % allow utf-8 input
\usepackage[T1]{fontenc}    % use 8-bit T1 fonts
\usepackage{hyperref}       % hyperlinks
\usepackage{url}            % simple URL typesetting
\usepackage{booktabs}       % professional-quality tables
\usepackage{amsmath}        % math environments (align, equation, etc.)
\usepackage{amsfonts}       % blackboard math symbols
\usepackage{amssymb}        % additional math symbols (circledast, etc.)
\usepackage{nicefrac}       % compact symbols for 1/2, etc.
\usepackage{microtype}      % microtypography
\usepackage{lipsum}		% Can be removed after putting your text content
\usepackage{graphicx}
\usepackage[style=apa,natbib=true]{biblatex}
\usepackage{doi}

\usepackage{placeins}  % Add to preamble
\usepackage{tabularray}
\usepackage{caption}
\usepackage{subcaption}
\DeclareMathOperator*{\argmin}{argmin}
\usepackage{blindtext}
\usepackage{xcolor}
\usepackage{bm}
\newcommand{\Conv}{%
  \mathop{\scalebox{1.5}{\raisebox{-0.2ex}{$\circledast$}}
  }
}
\usepackage{float}
\usepackage{amsthm}
\newtheorem*{theorem}{Theorem}
\newtheorem*{lemma}{Lemma}
\usepackage{booktabs} 
\usepackage{makecell}
\usepackage{hyperref,lipsum} 
\usepackage[normalem]{ulem} 

\usepackage{lineno}
\usepackage{csquotes}

\title{Cauchy-Gaussian Overbound for Heavy-tailed GNSS Measurement Errors}
\author{ \href{https://orcid.org/0009-0007-7755-8714}{\includegraphics[scale=0.06]{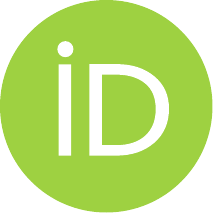}\hspace{1mm}Zhengdao Li} \\
	Department of Aeronautical and Aviation Engineering\\
	The Hong Kong Polytechnic University\\
	Hong Kong, China \\
	\texttt{zhengda0.li@connect.polyu.hk} \\
	\And
	\href{https://orcid.org/0000-0003-0994-9258}{\includegraphics[scale=0.06]{orcid.pdf}\hspace{1mm}Penggao Yan}\thanks{Correspondence author} \\
	Department of Aeronautical and Aviation Engineering\\
	The Hong Kong Polytechnic University\\
	Hong Kong, China \\
	\texttt{peng-gao.yan@connect.polyu.hk} \\
	\And
	\href{https://orcid.org/0000-0003-4158-0913}{\includegraphics[scale=0.06]{orcid.pdf}\hspace{1mm}Weisong Wen} \\
	Department of Aeronautical and Aviation Engineering\\
	The Hong Kong Polytechnic University\\
	Hong Kong, China \\
	\texttt{welson.wen@polyu.edu.hk} \\
	\And
	\href{https://orcid.org/0000-0002-0352-741X}{\includegraphics[scale=0.06]{orcid.pdf}\hspace{1mm}Li-Ta Hsu} \\
	Department of Aeronautical and Aviation Engineering\\
	The Hong Kong Polytechnic University\\
	Hong Kong, China \\
	\texttt{lt.hsu@polyu.edu.hk} \\
}

\renewcommand{\shorttitle}{Cauchy-Gaussian Overbound for Heavy-tailed Errors}

\hypersetup{
pdftitle={Cauchy-Gaussian Overbound for Heavy-tailed Errors},
pdfsubject={q-bio.NC, q-bio.QM},
pdfauthor={Zhengdao Li, Penggao Yan, Weisong Wen, Li-Ta Hsu},
pdfkeywords={Overbounding techniques, heavy-tailed distributions, Cauchy distribution, global navigation satellite systems},
}

\begin{document}
\maketitle

\begin{abstract}
	Overbounds of heavy-tailed measurement errors are essential to meet stringent navigation requirements in integrity monitoring applications. This paper proposes to leverage the bounding sharpness of the Cauchy distribution in the core and the Gaussian distribution in the tails to tightly bound heavy-tailed GNSS measurement errors. We develop a procedure to determine the overbounding parameters for both symmetric unimodal (s.u.) and not symmetric unimodal (n.s.u.) heavy-tailed errors and prove that the overbounding property is preserved through convolution. The experiment results on both simulated and real-world datasets reveal that our method can sharply bound heavy-tailed errors at both core and tail regions. In the position domain, the proposed method reduces the average vertical protection level by 15\% for s.u. heavy-tailed errors compared to the single-CDF Gaussian overbound, and by 21\% to 47\% for n.s.u. heavy-tailed errors compared to the Navigation Discrete ENvelope and two-step Gaussian overbounds.
\end{abstract}

% keywords can be removed
% \keywords{First keyword \and Second keyword \and More}

\vspace{0.5\baselineskip}

\section{Introduction} \label{sec: Introduction}
The reliability and safety of positioning solutions have become a critical concern for the {global navigation satellite systems (GNSS)}. Correspondingly, safety-of-life systems are established, such as satellite-based augmentation systems (SBAS), ground-based augmentation systems (GBAS), and receiver autonomous integrity monitoring (RAIM) \citep{unitedstates.dept.oftransportation_1999_2000, federalaviationadministration_gbas_2010, brown_baseline_1992, walter_weighted_1995}. These systems are designed to facilitate highly robust and precise GNSS positioning solutions, and maintain the integrity risk within the probability of hazardously misleading information events ($P_{HMI}$). As an alternative to integrity risk evaluation, the protection level (PL) is derived in the position domain to measure the maximum tolerable positioning error boundary under $P_{HMI}$ \citep{blanch_gaussian_2018, elsayed_bounding_2023, antic_sbas_2023}. 
% {A lower PL indicates a less conservative positioning error bound without raised alarm} and typically suggests a higher level of system availability \citep{hassan_review_2021}. 

To ensure the integrity of navigation systems, {it is crucial to precisely characterize the measurement error, as the stochastic properties of the measurement error will be projected into the position domain \citep{lee_sigma_2009}}. However, the limited bandwidth channels of GNSS augmentation systems hinder the transmission of complicated error profiles to users \citep{blanch_gaussian_2018}. Therefore, a simpler and more conservative error boundary (namely, overbound) is commonly employed in navigation systems \citep{blanch_gaussian_2018, rife_overbounding_2004, blanch_error_2005,
rife_paired_2006, gao_error_2022}, which compromises between accurate error modelling and the complexity of broadcasting larger amount of parameters. Overbounds represent the worst error distribution without hardware failures and are designed to guarantee that the integrity risk remains below an acceptable level \citep{decleene_defining_2000, rife_core_2004, antic_sbas_2023,xia_integrityconstrained_2024}. A sharp (i.e., closely matching the original error distribution) overbound helps reduce PL and facilitates the development of high-availability integrity monitoring algorithms. Moreover, the overbound should also be explicitly parameterized to be conveniently broadcast to users \citep{blanch_gaussian_2018, blanch_baseline_2015}.

Gaussian overbounding techniques have been a common approach for characterizing GNSS error distributions. The basic single cumulative density function (CDF) Gaussian overbound \citep{decleene_defining_2000} assumes strictly symmetric unimodal (s.u.) errors, a limitation addressed by paired Gaussian overbound \citep{rife_paired_2006} and more advanced methods like two-step Gaussian overbound \citep{blanch_gaussian_2018}. The classic Gaussian-based methods established the framework for CDF overbound, as detailed in Section \ref{sec: classic gaussian ob}. However, a persistent challenge is the conservative performance of Gaussian overbounds when applied to the error distributions with heavy-tailed properties \citep{blanch_gaussian_2018, yan_principal_2024, huang_robust_2016}, which are common in ephemeris/clock errors and multipath errors \citep{heng_statistical_2011,foss_introduction_2013, zhu_gnss_2018, karaim_gnss_2018,yan2025high}. {It has been verified that typical heavy-tailed GNSS error distributions include large errors (beyond 2-3 standard deviations) that occur with a higher-than-Gaussian frequency, even though the core region may still be well-represented by a Gaussian shape \citep{rife_overbounding_2012}. When visualized as a probability density function (PDF), a heavy-tailed distribution exhibits tails that decay more slowly than a Gaussian, resulting in a sharper peak and prolonged tails. Generally, the ``heavy-tailedness" of a distribution is often measured by its kurtosis \citep{feldmann_fitting_1998, balanda_kurtosis_1988}. Regarding the properties of heavy-tailed errors, Gaussian overbounds employ a large sigma ($\sigma$). While this may successfully bound the tails, it often creates a significant separation between the overbound and the empirical data in the core region, introducing unnecessary conservatism.}

Non-Gaussian overbounding methods have been developed to address the challenges posed by heavy-tailed errors. \citet{rife_core_2004} proposed the core overbounding method, which uses a Gaussian core Gaussian sidelobes (GCGS) distribution to provide less conservative bounding on the core and tail regions separately. \citet{blanch_position_2008} proposed a {zero-mean} bimodal Gaussian mixture model (BGMM) based method to add a heavy-tailed Gaussian component in constructing the overbounding distribution. \citet{rife_overbounding_2012} proposed the Navigation Discrete ENvelope (NavDEN), a symmetric and discrete overbound with a tight Gaussian core and flared tails designed to model heavy-tailed distributions. {However, the low grid resolution of the coarse discrete models may limit the bounding performance.} \citet{xue_upper_2017} adopted the stable distribution for bounding GBAS ranging errors and validated the feasibility using simulated error samples. \citet{larson_gaussianpareto_2019} developed the Gaussian-Pareto overbound, which applies Extreme Value Theory to generate a sharp tail bound. {However, the Gaussian-Pareto overbound does not inherently satisfy the s.u. requirement defined by \citet{decleene_defining_2000}. Consequently, the overbound cannot be analytically proven to preserve overbounding properties through convolution, leaving its performance in the position domain unverified.} Recently, \citet{yan_principal_2024} proposed the Principal Gaussian overbound (PGO) to tightly bound original errors in both the core and tail regions by strategically inflating and shifting the Gaussian components of a fitted BGMM for the error distribution. Although the PGO demonstrates promising bounding performance for heavy-tailed errors, the PGO is developed based on the zero-mean assumption on the error distribution, which is not usually satisfied in real-world applications.

% \textcolor{red}{However, the Gaussian-Pareto model does not strictly have a s.u. shape and fails to preserve overbounding properties through convolution, leaving its bounding performance in the position domain unverified.}

% Despite significant advancements in bounding heavy-tailed distributions, the aforementioned non-Gaussian overbounding methods are either susceptible to concrete grid resolution, or validated solely through simulations or are applicable only to zero-mean s.u. errors.
% {While the aforementioned non-Gaussian overbounds represent significant advancements in bounding heavy-tailed distributions, they may be limited by coarse discrete models with low grid resolution, a validation based purely on simulations, or an inability to handle not symmetric unimodal (n.s.u.) errors.}  % commented out on 2025-12-22

In this work, we aim to develop a systematic method to bound heavy-tailed error distributions, including both s.u. profiles with a zero bias and n.s.u. profiles with a non-zero bias. The core overbounding idea \citep{rife_core_2004} is used to separately bound the core and tail regions of the empirical errors distribution. Unlike conventional non-Gaussian overbounding methods, which employ non-Gaussian distributions for tail bounding and Gaussian distributions for core bounding, we propose leveraging the intrinsically heavy-tailed properties of the Cauchy distribution to sharply bound the core region of empirical errors. This Cauchy core is then transitioned to the Gaussian distribution to tightly bound the tails of the empirical errors. The Cauchy distribution (Lorentzian in physics), was first applied to describe spectral lines due to homogeneous broadening \citep{cauchy_exercices_1840, born_principles_2013}. The naturally heavy-tailed property and the explicit forms of PDF and CDF have made Cauchy applied in various areas, including catastrophe predictions in computational finance \citep{mahdizadeh_goodnessoffit_2019}, rainfall probability model in Hydrology \citep{kassem_identification_2021}, and loss function in machine learning \citep{liu_cauchy_2022}. We found that the heavy-tailed Cauchy distribution possesses significantly more prolonged tail regions and thus a sharper core region than the Gaussian distribution. This property allows the Cauchy-based overbounds to closely adhere to the core of empirical error distributions with heavy tails while diverging from the tail. In contrast, the Gaussian distribution allocates a considerably higher proportion of cumulative mass at its core, and this mass exponentially decreases when the magnitude of error gradually rises, resulting in short tail regions. Hence, the Gaussian overbound can sharply bound the tails of empirical errors but leave large separations away from the core. {Additionally, following the core overbounding framework \citep{rife_core_2004}, we seek to balance between bound sharpness and operational simplicity by separately generating bounds at core and tail regions.} Therefore, naturally, we have the motivation to combine the strengths of the Cauchy distribution in the core region with those of the Gaussian distribution in the tail regions to bound heavy-tailed error distributions. This idea constitutes the main content of our work. 
% Besides, we are inspired by the core overbounding approach proposed by \citep{rife_core_2004} includes overbound at core and tail, respectively, to achieve bound sharpness and operational simplicity. 

This paper proposes the Cauchy-Gaussian overbound for both s.u. and n.s.u. error distributions, where an explicit framework is developed to determine the optimal parameters of these overbounds. Inspired by the Gaussian overbounding methods \citep{decleene_defining_2000, rife_paired_2006, blanch_gaussian_2018}, Cauchy-Gaussian overbound is designed to function as a single-CDF bound for s.u. errors and as paired bounds for n.s.u. errors. We benchmark the bounding performance of the proposed method with single-CDF Gaussian overbound \citep{decleene_defining_2000} for s.u. errors using a simulated dataset, and with the two-step Gaussian overbound \citep{blanch_gaussian_2018} for n.s.u. errors using a real-world dataset, respectively. Results show that the proposed Cauchy-Gaussian overbound can more tightly bound the heavy-tailed errors, at both core and tail regions. {In the position domain, the proposed method can reduce the vertical protection level (VPL) by roughly 15\% on average for s.u. errors. For n.s.u. errors, the improvement is greater, with average VPL reductions of 21\% and 47\% compared to the NavDEN and two-step Gaussian overbounds, respectively.}

The contributions of this work lie in three aspects:
\begin{enumerate}
    \item Propose the Cauchy-Gaussian CDF overbound for both s.u. and n.s.u. heavy-tailed error distributions. A procedure to determine the overbound parameters is developed; 
    \item Prove that the overbounding properties of the proposed overbound is preserved through convolution, which provides theoretical support for using Cauchy-Gaussian CDF overbound in integrity applications; 
    \item Validate the bounding performance of the proposed method in both range and position domains for both s.u. and n.s.u. error distributions, using simulated and real datasets. 
\end{enumerate}

The remaining parts of this paper are organized as follows: Section \ref{sec: classic gaussian ob} reviews three classic Gaussian overbounding methods. Section \ref{sec: proposed ovb} develops the Cauchy-Gaussian overbound for both s.u. and n.s.u. errors. Section \ref{sec: experiments and evaluations} compares the bounding performance of Cauchy-Gaussian overbound, single-CDF Gaussian overbound, and two-step Gaussian overbound in both range and position domains through numerical simulations and real datasets validations. The impacts of heavy-tailedness of the empirical errors with non-zero bias on bounding performance are investigated in Section \ref{sec: analyze heavy-tailedness}. Finally, Section \ref{sec: conclusion and outlook} concludes this study and proposes prospective research directions.

\section{Review on Gaussian Overbounding Methods} \label{sec: classic gaussian ob}
The Gaussian distribution has been widely applied in overbounding methods, due to its simple parameterization and property of preserving overbounding nature through convolution \citep{rife_paired_2006, yan_principal_2024,larson_gaussianpareto_2019}. In this section, three classic Gaussian overbounding methods are reviewed, including (1) single-CDF Gaussian overbound \citep{decleene_defining_2000}; (2) paired Gaussian overbound \citep{rife_paired_2006}; and (3) two-step Gaussian overbound \citep{blanch_gaussian_2018}. The paired Gaussian overbound has inspired the proposal of the Cauchy-Gaussian overbound in this paper, while the Gaussian single-CDF overbound and the two-step Gaussian overbound will serve as the baseline in evaluating the bounding performance in Section \ref{sec: experiments and evaluations}. For the sake of notation, $F$ indicates CDF and $f$ indicates PDF in the remainder of this paper.

\subsection{Single-CDF Gaussian overbound} \label{sec: sym cdf ovb}
\citet{decleene_defining_2000} introduced the first successful CDF overbound on s.u. distributions. {Specifically, the CDF overbound for s.u. errors is given by:
\begin{align}
    F_{ob}(x)   \geq F_e(x)  \quad\forall F_e\leq \frac{1}{2}, \label{equ: sym cdf x<0}\\
    F_{ob}(x)  \leq F_e(x) \quad\forall F_e > \frac{1}{2}, \label{equ: sym cdf x>0}
\end{align}
where $F_{ob}$ and $F_e$ denote the CDFs of the overbounding distribution and the empirical error distribution, and both distributions are assumed to have a zero mean. When $F_{ob}$ is represented by the CDF of a Gaussian distribution}, the tightest single-CDF Gaussian overbound can be determined by finding the minimum scale parameter that satisfies the above inequalities. Decleene proved that if both empirical errors and overbound are s.u. distributions, the projection of the range-domain overbound to the position domain also bounds the positioning errors, which mathematically means the overbounding properties are preserved through convolution \citep{decleene_defining_2000}. Despite laying the foundation of CDF overbound, Decleene's method strictly relies on s.u. error and overbounding distributions and is inadequate for real-world GNSS measurements that commonly exhibit asymmetry, multi-modality, and biases \citep{blanch_gaussian_2018, rife_paired_2006}.

\subsection{Paired Gaussian overbound}
\label{sec: paired ovb}
The paired overbounding method expands the application scope of Decleene's method from bounding s.u. errors to random error profiles, by introducing a pair of bounds \citep{rife_paired_2006}. Specifically, the left and right (indicated respectively with subscript $L$ and $R$) Gaussian overbounds preservable through convolution are given by:
\begin{align}
    F_L(x) = F_G(x; \mu_L,\sigma_L)  &\geq F_e(x)\quad\forall x, \label{equ: paired L} \\
    F_R(x) = F_G(x; \mu_R,\sigma_R)  &\leq F_e(x)\quad\forall x, \label{equ: paired R}
\end{align}
where $F_G$ represents Gaussian CDF, $\mu$ and $\sigma$ denote the location and scale parameters \footnote{To maintain generality across both Gaussian and Cauchy distributions, we will use the terms ``location parameter" and ``scale parameter" for the remainder of the paper.} (i.e., mean and standard deviation) of the Gaussian distributions, and $\mu_L=-\mu_R$. Following the definition in Equation \eqref{equ: sym cdf x<0} and \eqref{equ: sym cdf x>0}, an analog single-CDF overbound $F_{ob}$ based on the paired overbound can be constructed by:
\begin{equation}
    F_{ob}(x)= \left\{
        \begin{array}{lr} 
          F_{L}\left(x \right)        &  \forall F_{L} < \frac{1}{2} \\
         \frac{1}{2}                             &  \text{otherwise} \\
          F_{R}\left(x \right)        &  \forall F_{R}  > \frac{1}{2}
        \end{array}.
    \right. \label{equ: single-cdf analog}
\end{equation}
Notably, the formed $F_{ob}$ is only used for visualization purposes and will not be used in calculating the projection of the range-domain overbound to the position domain. Overbound using paired distributions effectively tackles non-symmetric or non-unimodal empirical errors. However, finding the paired bounds using Gaussian distribution typically results in large location parameters ($\mu_L$ and $\mu_R$), which unavoidably increase the conservatism in bounding the measurement error and eventually raises the PL in the position domain.

\subsection{Two-step Gaussian overbound}
\label{sec: twostep gaussian ovb}
Based on the paired overbounding method, \citet{blanch_gaussian_2018} proposed a state-of-the-art two-step Gaussian overbound to further tighten the overbounding distribution towards the empirical profile. This is achieved through an ad hoc s.u. distribution $f_{su}(t;b_{su})$, with a non-zero location parameter $b_{su}$. The $f_{su}$ performs as an intermediate paired left and right CDF bounds, according to the inequalities in Equation \eqref{equ: paired L} and \eqref{equ: paired R},
\begin{align}
    \int_{-\infty}^{x} f_{su_L}(t; b_{su_L})dt \geq F_e(x)\quad\forall x  , \label{equ: 2step1 L}\\
    \int_{-\infty}^{x} f_{su_R}(t; b_{su_R})dt \leq F_e(x)\quad\forall x  . \label{equ: 2step1 R}
\end{align}
Given that both $f_{su}$ and the Gaussian distribution are s.u., the second step aims to find the minimum scale parameters $\sigma_L$ and $\sigma_R$ on the left and right regions to satisfy:
\begin{align}
    F_G(x; b_{su_L}, \sigma_L)  &\geq   \int_{-\infty}^{x} f_{su_L}(t; b_{su_L})dt\quad\forall x \leq b_{su_L}, \label{equ: 2step2 L}  \\
     F_G(x; b_{su_R}, \sigma_R)  &\leq   \int_{-\infty}^{x} f_{su_R}(t; b_{su_R})dt\quad\forall x \geq b_{su_R}.  \label{equ: 2step2 R}
\end{align}
Finally, the left and right Gaussian CDF bounds are obtained by:
\begin{align}
    F_L(x) &=  F_G(x; -b_{f}, \sigma_f)\quad\forall x \leq -b_{f}, \\
      F_R(x) &=  F_G(x; b_{f}, \sigma_f)\quad\forall x \geq b_{f},
\end{align}
where $b_{f}=\max\left( |b_{su_L}|, |b_{su_R}|\right)$ and $\sigma_f=\max\left( \sigma_L, \sigma_R\right)$. Besides, the two-step Gaussian overbounding method maintains the overbounding properties through convolution, given that the inequalities in Equation \eqref{equ: 2step1 L} to \eqref{equ: 2step2 R} are satisfied. Compared to the conventional paired Gaussian overbound in Section \ref{sec: paired ovb}, the two-step Gaussian overbounding method yields a significantly smaller location parameter $b_{f}$ than the paired Gaussian overbound, while not increasing the number of parameters (i.e., location and scale) in the overbounding distributions to be broadcast. However, the involvement of the intermediate $f_{su}$ can still make the error bounds fairly conservative. For heavy-tailed and biased empirical errors, the two-step Gaussian overbound would still result in either a large $b_f$ or a large $\sigma_f$, which increases the PL in the position domain to maintain integrity. 

\section{Cauchy-Gaussian Overbound}
\label{sec: proposed ovb}
The intrinsically heavy-tailed properties of the Cauchy distribution can result in significantly prolonged tail regions and a sharp core region, making the Cauchy distribution a fit to bound the core of empirical errors. In contrast, the Gaussian distribution possesses light tails and can potentially bound errors more tightly at tail regions. Based on the properties, this section introduces a three-step procedure to construct the Cauchy-Gaussian overbound for symmetric unimodal (s.u.) and not symmetric unimodal (n.s.u.) error distributions separately. As shown in the workflow of Figure \ref{fig: flowchart su}, the optimal single-CDF overbounds using zero-mean Cauchy and Gaussian distributions, respectively, are first determined in Step 1 for the empirical s.u. errors without bias. {The \enquote{optimality} particularly refers to the case where the minimum separation between the overbound and the empirical errors is reached, given a specified overbounding framework (e.g., single-CDF Cauchy overbound or single-CDF Gaussian overbound in this context). This specific definition will be used consistently in the subsequent analysis.} Processed by tangential transitions in Step 2, the two formed single-CDF overbounds are integrated in Step 3 to generate the single-CDF Cauchy-Gaussian overbound. In Figure \ref{fig: flowchart nsu},  the original n.s.u. error distribution is optimally bounded with paired Cauchy-Gaussian Combined Model (CGCM) overbound in Step 1 and paired Gaussian overbound in Step 2. Finally, in Step 3, the right half of the synthesized paired Cauchy-Gaussian overbound is a point-wise supremum of the right-half paired CGCM overbound and right-half paired Gaussian overbound, while the left half of the proposed overbound features a point-wise infimum of the left half of the two prior overbounds. Notably, the $\Phi^{-1}(\cdot)$ transforms CDFs to the equivalent standard normal quantiles, thus Gaussian-based segments in overbounding distributions perform as straight lines. We start with a short introduction to the Cauchy distribution.
\begin{figure}[ht]
        \centering
        \begin{subfigure}{0.49\textwidth} % width for the subfigure
            \centering
            \includegraphics[width=0.85\textwidth]{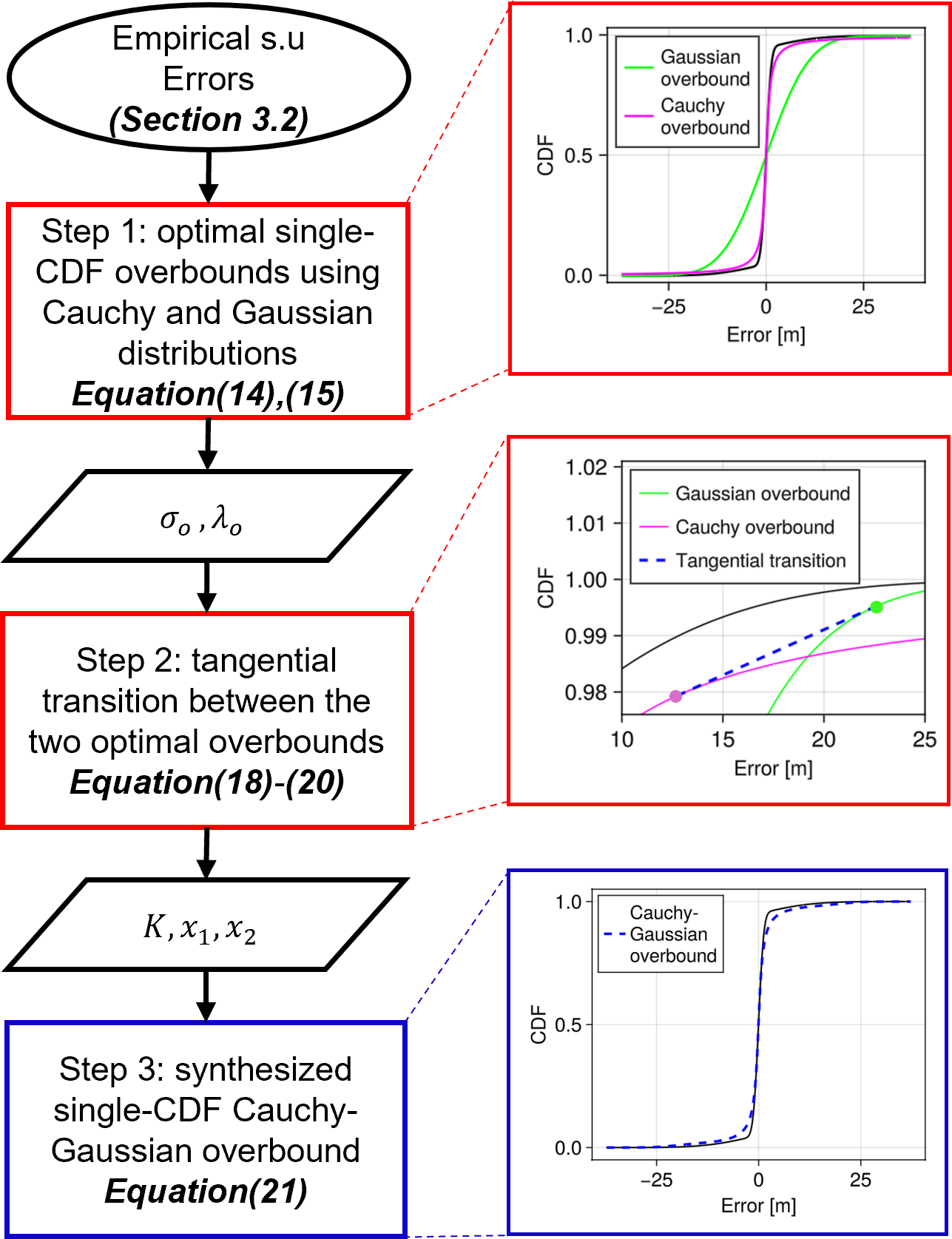}
            \caption{}
            \label{fig: flowchart su}
        \end{subfigure}
        \hfill 
        \begin{subfigure}{0.49\textwidth} 
            \centering
            \includegraphics[width=0.95\textwidth]{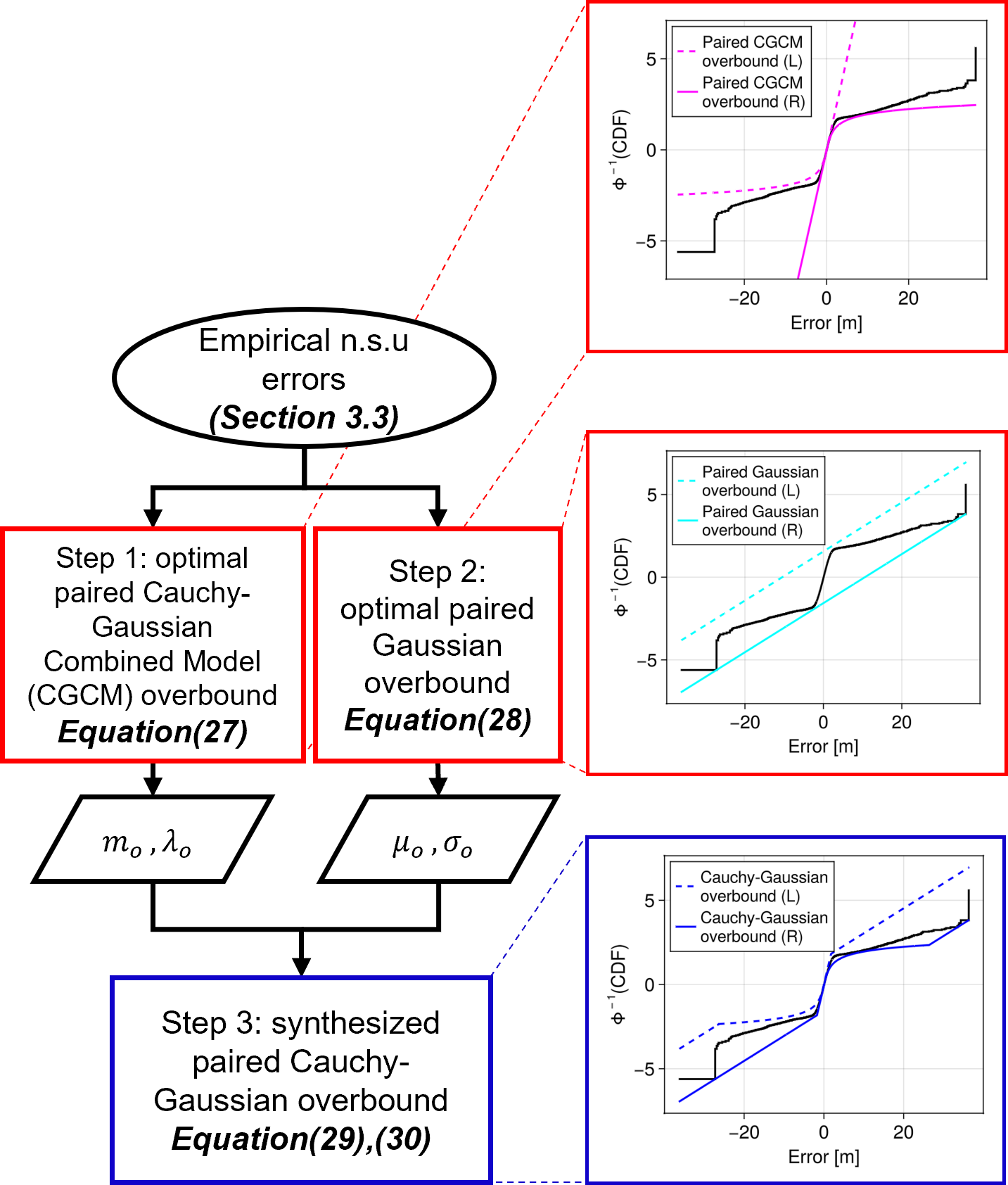}
            \caption{}
            \label{fig: flowchart nsu}
        \end{subfigure}
        \caption{Flowcharts of the Cauchy-Gaussian overbound for (a) s.u. and (b) n.s.u. error distributions. {The black curve in each subfigure represents the empirical error distribution. For better visualization at tail regions, the subfigures in (b) show the CDF values transformed into their equivalent standard normal quantiles.}}
        \label{fig: flowchart}
\end{figure}
\subsection{Cauchy distribution}
\label{sec: cauchy distribution}
Similar to the Gaussian distribution, the PDF of a Cauchy distribution ($C(m,\lambda)$) is a symmetric \enquote{bell curve}. The Cauchy distribution has explicit forms of both PDF ($f_C$) and CDF ($F_C$) expressed with a location parameter $m$ and scale parameter $\lambda$ as follows: 
\begin{align}
    f_C(x; m, \lambda) &=\frac{1}{\lambda \pi} \cdot \frac{1}{1+(\frac{x-m}{\lambda})^2}, \\
    F_C(x; m, \lambda) &=\frac{1}{2}+\frac{1}{\pi}\arctan \left( \frac{x-m}{\lambda} \right). 
\end{align}
\begin{figure}[t]
        \centering
        \begin{subfigure}{0.49\textwidth} % width for the subfigure
			% \centering
            \includegraphics[width=0.95\textwidth]{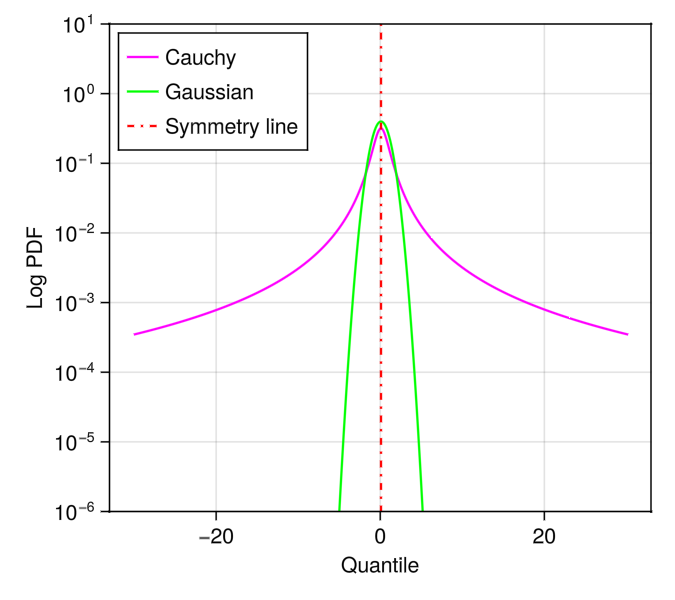}
            \caption{\centering}
            \label{fig: logpdf example}
        \end{subfigure}
        \hfill
        \begin{subfigure}{0.49\textwidth} 
			% \centering
            \includegraphics[width=0.95\textwidth]{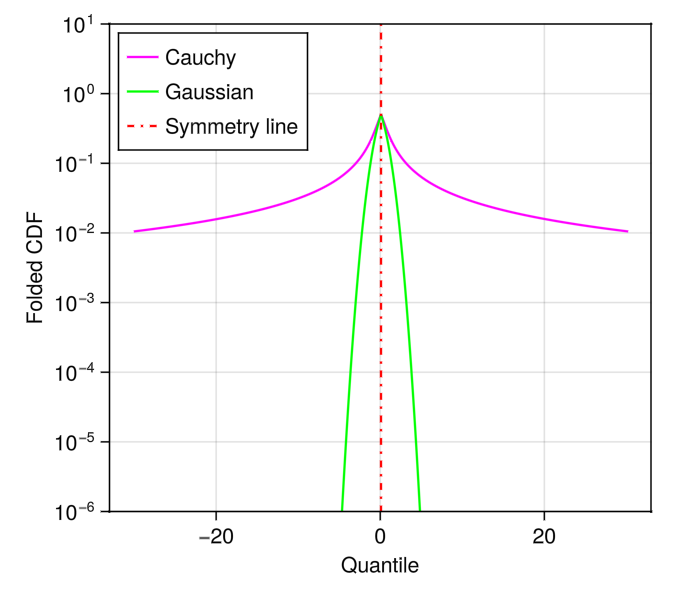}
            \caption{\centering}
            \label{fig: foldedcdf example}
        \end{subfigure}
        \caption{{Comparison between the standard Cauchy and Gaussian distributions, through (a) PDF and (b) folded CDF on a logarithmic scale.}}
        \label{fig: demo C G PDF comparisons}
\end{figure}
The differences between Cauchy and Gaussian distribution are explicitly reflected in their PDFs \citep{menon_characterization_1962, smelser_international_2001}, where the Gaussian PDF decreases at an exponential rate with the increase of error magnitude and the Cauchy PDF decreases at a polynomial rate. Figure \ref{fig: logpdf example} shows the logarithmic-scale (Log) PDF of the standard Cauchy and Gaussian distributions, with the location parameter being 0 and the scale parameter being 1 for both distributions. The inset plot illustrates that the standard Cauchy distribution has a sharper peak at the core region (near the symmetry line) than the standard Gaussian distribution. Moreover, the standard Cauchy distribution exhibits considerably longer and fatter tails than the standard Gaussian distribution. Correspondingly, Figure \ref{fig: foldedcdf example} shows the folded CDFs of the standard Cauchy and Gaussian distributions. The folded CDF combines the left half of CDF and right half of complementary CDF (CCDF), i.e., 1-CDF, in one view. As can be seen, the Cauchy distribution allocates a higher proportion of accumulated mass at tail regions than the Gaussian distribution does.\par

The extremely heavy-tailed properties of the Cauchy distribution place an advantage in terms of bounding heavy-tailed empirical errors. Specifically, the Cauchy distribution decentralizes a higher percentage of probability mass to the tails, which makes it possible to bound heavy-tailed empirical errors without over-inflating the scale parameter. Besides, the Cauchy distribution features a significantly sharper core region than the Gaussian. This property allows the Cauchy overbound to closely adhere to the core of empirical error distributions. 

Nevertheless, the extremely heavy-tailed properties of the Cauchy distribution become even more pronounced after multiple rounds of convolution. Therefore, the resultant positioning error distribution from the convolution of multiple Cauchy-characterized ranging error sources is expected to exhibit significantly heavy tails, which are usually associated with a large PL in the position domain. In contrast, the Gaussian distribution decreases exponentially in PDF, guaranteeing a small probability mass in the tail region of itself and its convolutions. To achieve tighter bounding in both the core and tail regions in the range domain and thereby reduce the PL in the position domain, it is necessary to combine the strengths of the Cauchy distribution in the core region and the Gaussian distribution in the tail regions. This is the central concept behind the proposed Cauchy-Gaussian overbound for heavy-tailed distributions, as illustrated below.

\subsection{Procedures of bounding s.u. errors}
\label{sec: method su errors}
As mentioned in Section \ref{sec: sym cdf ovb}, for a zero-mean s.u. empirical errors, its single-CDF overbound falls on finding the s.u. distribution (i.e., Cauchy or Gaussian distribution in this study) with the minimum scale parameter. Based on the idea of utilizing the advantages of both Cauchy and Gaussian distributions, we propose the Cauchy-Gaussian overbound with a three-step procedure as follows:

\textbf{Step 1: Overbound with optimal Cauchy and Gaussian distributions} \\
Single-CDF overbound on s.u. errors, standing to inequalities in Equation \eqref{equ: sym cdf x<0} and \eqref{equ: sym cdf x>0}, is conducted using both Gaussian ($G(0,\sigma)$) and Cauchy ($C(0,\lambda)$) distributions, respectively. Denote the optimal scale parameter of the two distributions as $\sigma_o$ and $\lambda_o$, respectively. For the single-CDF Gaussian overbound, $\sigma_o$ can be found by:
\begin{subequations}
\begin{align}
        &\sigma_o = \min   \quad  \sigma\\
s.t.\quad \quad  & F_{G}(x;0,\sigma) \geq F_e(x)\quad \forall x \leq 0 \\
         & F_{G}(x;0,\sigma) \leq F_e(x)\quad \forall x > 0. 
\end{align}
\label{opt1: getting optimal gaussian}
\end{subequations}
For the Cauchy distribution, we prove in Appendix \ref{app 1: cauchy ovb on a known gaussian} that the heavy-tailed Cauchy distribution can well overbound a centrally-aligned Gaussian distribution by $\lambda \geq \sqrt{\frac{2}{\pi}}\sigma$. This indicates that $\sigma_o$ determines the upper bound of $\lambda_o$, which decreases the computational load. Specifically, the objective of finding the optimal Cauchy distribution gives:
\begin{subequations}
\begin{align}
        &\lambda_o = \min   \quad \lambda \\
s.t. \quad\quad & F_C(x;0,\lambda) \geq  F_e(x)\quad \forall x \leq 0 \\
         & F_C(x;0,\lambda) \leq F_e(x)\quad \forall x > 0 \\
          & \lambda \leq \sqrt{\frac{2}{\pi}}\sigma_o.
\end{align}
\label{opt2: getting optimal cauchy}
\end{subequations}
Inspired by the work of \citet{blanch_position_2008},  {we demonstrate an instance of heavy-tailed empirical errors by a zero-mean BGMM with the following settings}:
\begin{equation}
    f_e(x)=0.9f_G(x;0,1)+0.1f_G(x;0,10) .
    \label{eq:example_error_type_I}
\end{equation}
The black curve in Figure \ref{fig: tentative combined overbound} shows this example error profile. The optimal Cauchy and Gaussian overbounds for this empirical error distribution are shown by magenta and green curves, respectively. As can be seen, the Gaussian overbound deviates more significantly from the empirical errors in the core region while aligning more closely with the empirical errors in the tail regions, compared to the Cauchy overbound. This observation further confirms our anticipation that Cauchy bounds more sharply at the core region and Gaussian yields tighter bounding at tails. Therefore, it is reasonable to combine the strengths of the two models to achieve tighter bounding.

\textbf{Step 2: Tangential transition to preserve overbounding properties through convolution} \\
Using the optimal Gaussian and Cauchy overbounds, we may combine the Cauchy core and Gaussian tails and tentatively define the combined CDF overbound as:
\begin{equation}
    F_{combined}(x)= \left\{
        \begin{array}{lr} 
          \min\left(F_C(x;0,\lambda_o), F_G(x;0,\sigma_o)\right)        &  \forall  x \leq 0 \\
          \max\left(F_C(x;0,\lambda_o),F_G(x;0,\sigma_o)\right)               &  \forall  x > 0
        \end{array}.
    \right.
    \label{equ: f_combined}
\end{equation}
The CDF of the combined distribution has been plotted as the blue dashed curve in Figure \ref{fig: xt cdf}. As can be seen, the Cauchy overbound intersects with the Gaussian overbound on both left and right tail regions. The abscissa of the intersections are denoted as $x_{t,L}$ and $x_{t,R}$, respectively. The zoom-in view in Figure \ref{fig: xt cdf} also illustrates that $F_{combined}$ transits from the Cauchy overbound directly to the Gaussian overbound at $x_{t,R}$ in the right region. A similar transition from the Gaussian overbound to the Cauchy overbound exists in the left region, which is omitted for the present. The transitions destroy the monotonicity of the PDF for each half of the combined distribution, as shown in the zoom-in view of Figure \ref{fig: xt pdf}. A sudden jump from Cauchy PDF to Gaussian PDF occurs at $x_{t,R}$. This phenomenon makes the combined distribution violate the s.u. property and do not preserve overbounding properties through convolutions\citep{decleene_defining_2000}. To solve this problem,  a tangential transition from Cauchy to Gaussian is constructed.

\begin{figure}[t]
        \centering
        \begin{subfigure}{0.49\textwidth} % width for the subfigure
            \includegraphics[width=0.95\textwidth]{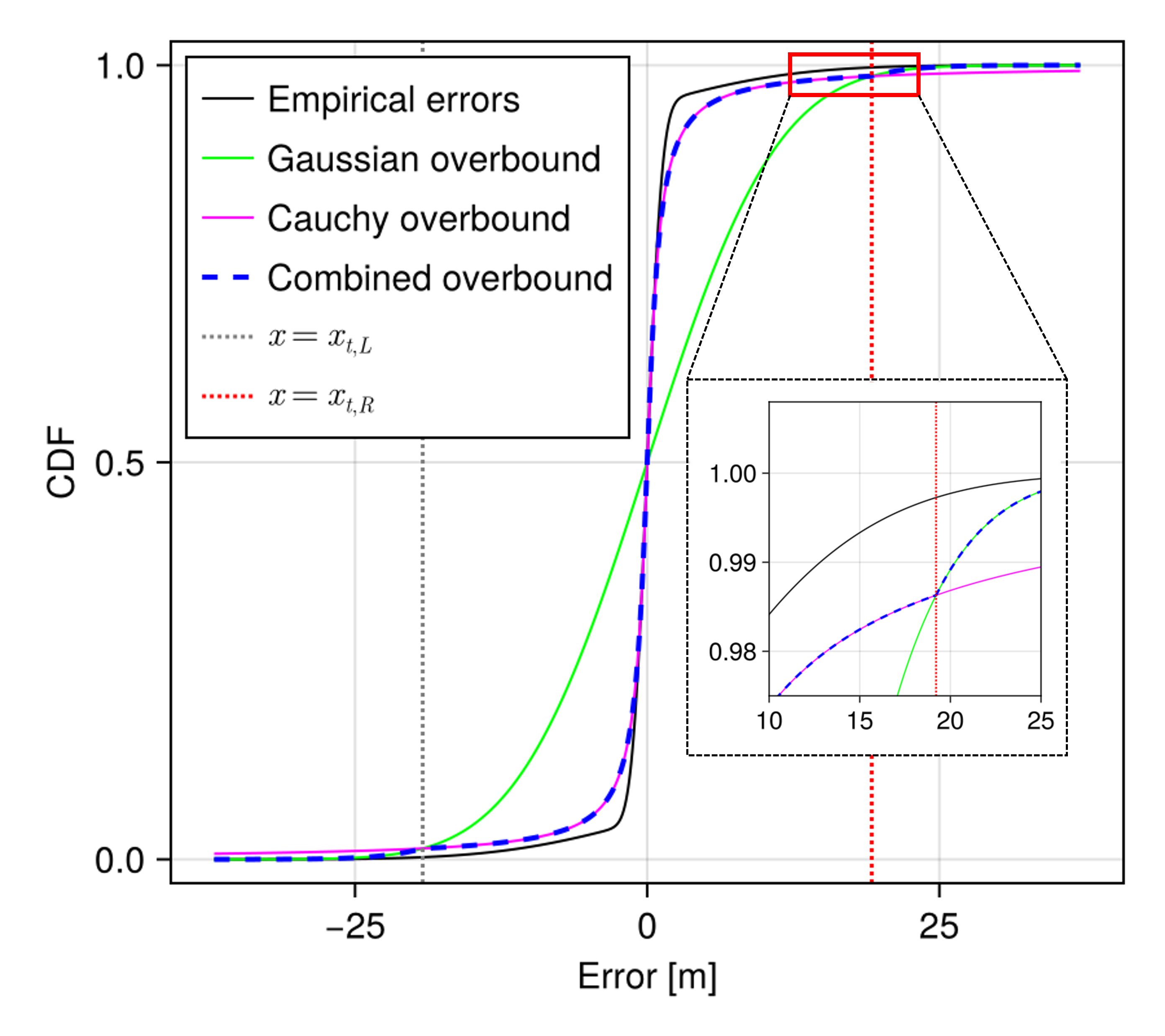}
            \caption{\centering}
            \label{fig: xt cdf}
        \end{subfigure}
        \hfill 
        \begin{subfigure}{0.49\textwidth} 
            \includegraphics[width=0.95\textwidth]{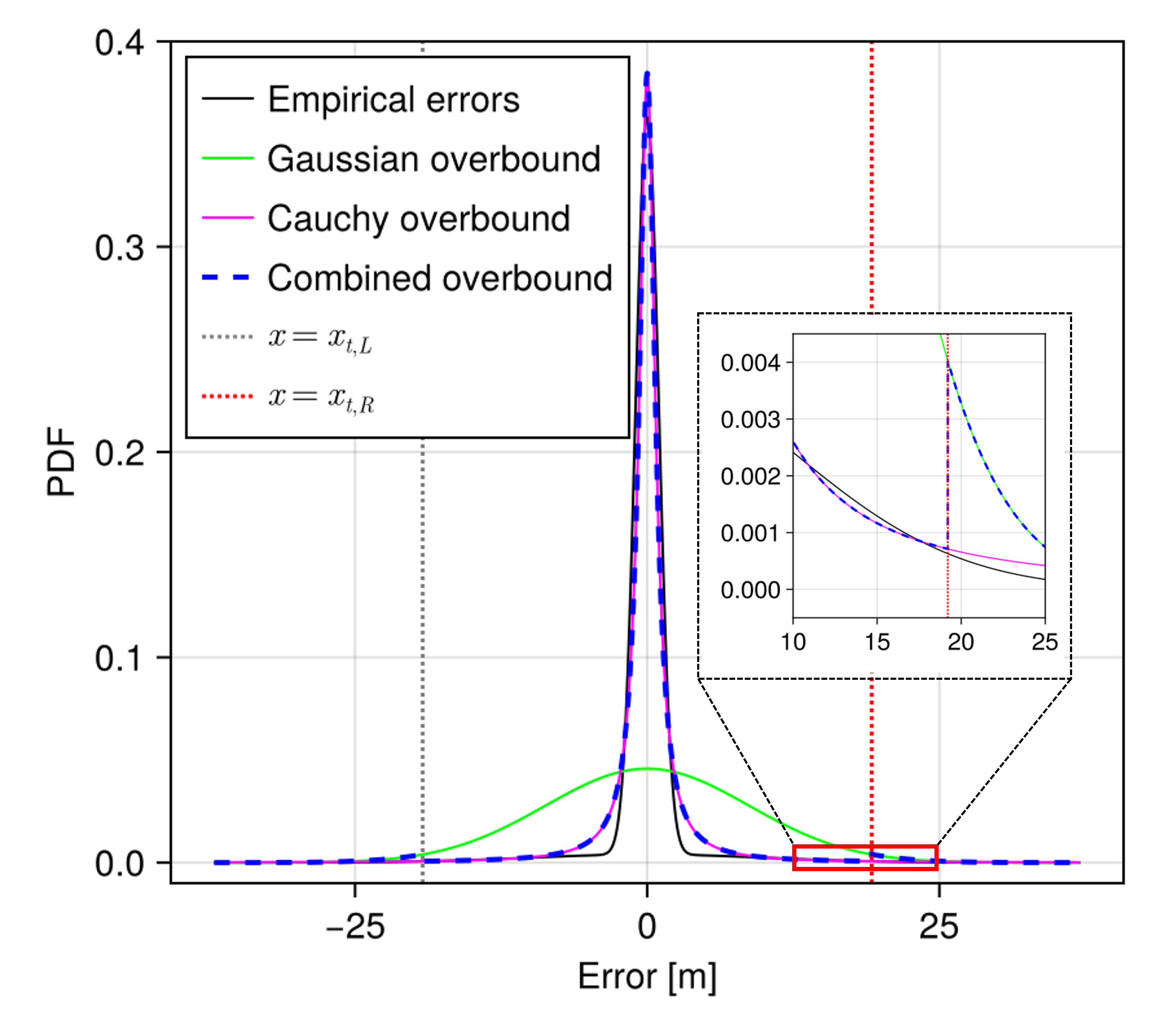}
            \caption{\centering}
            \label{fig: xt pdf}
        \end{subfigure}
        \caption{The tentative combined overbounding distribution for zero-mean BGMM error defined in Equation \eqref{eq:example_error_type_I} (s.u. profile) in two views: (a) CDF; (b) PDF. The Gaussian and Cauchy overbounds are also plotted for reference in each subfigure.}
        \label{fig: tentative combined overbound}
\end{figure}

\begin{figure}[t]
        \centering
        \begin{subfigure}{0.49\textwidth} % width for the subfigure
            \includegraphics[width=0.95\textwidth]{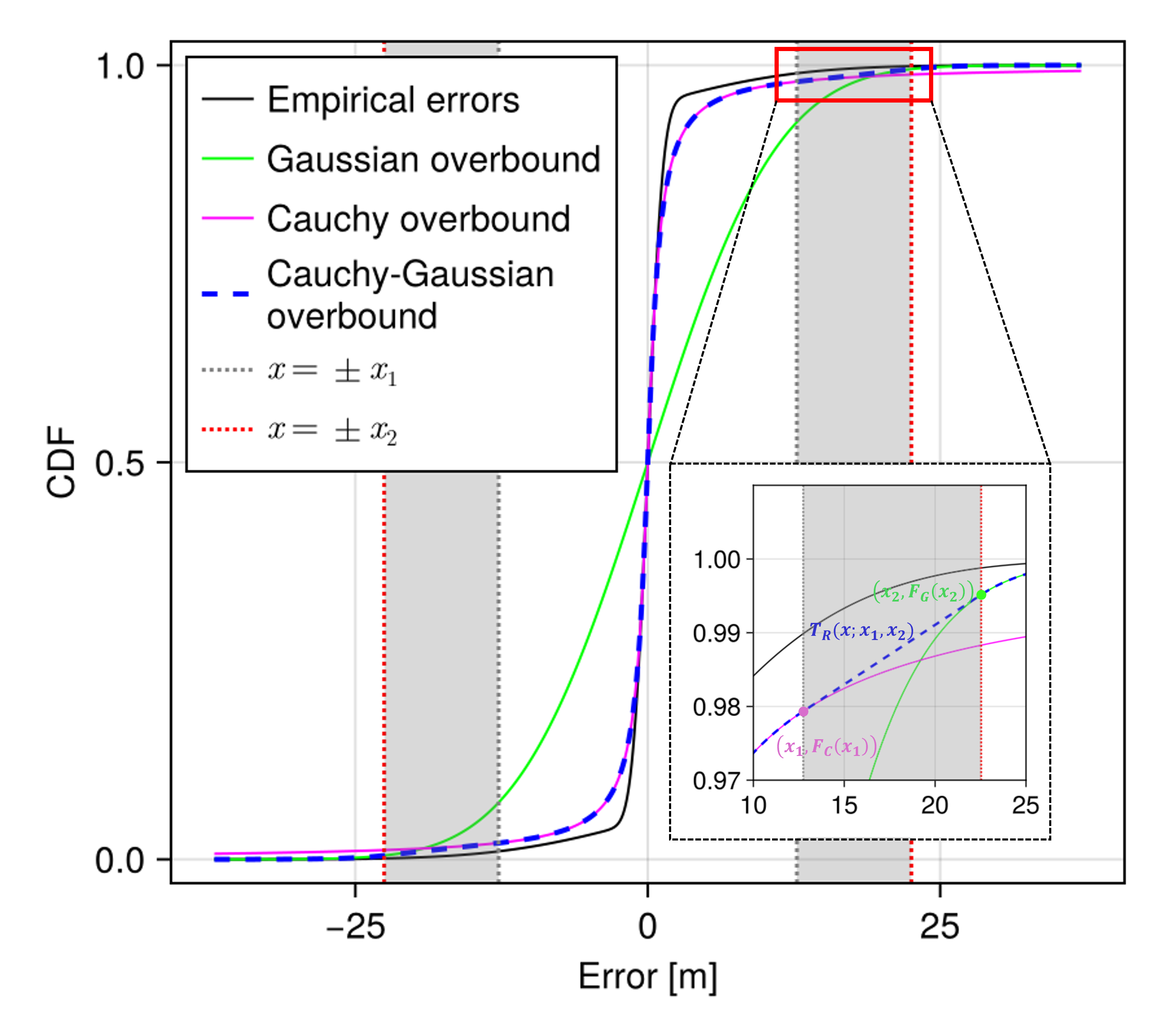}
            \caption{\centering}
            \label{fig: sym ovb cdf}
        \end{subfigure}
        \hfill 
        \begin{subfigure}{0.49\textwidth} 
            \includegraphics[width=0.95\textwidth]{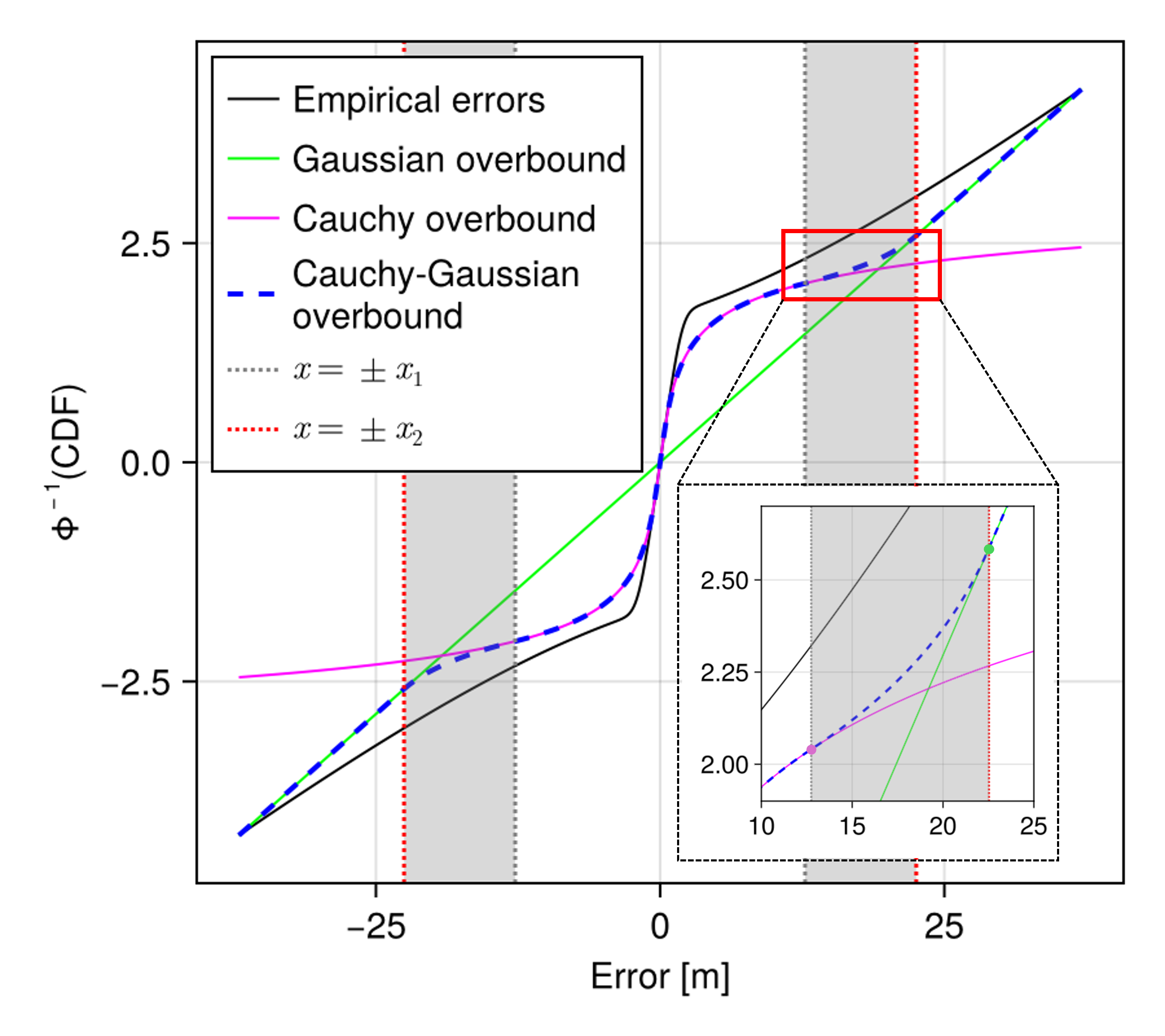}
            \caption{\centering}
            \label{fig: sym ovb normcdf}
        \end{subfigure}
        \\
        \begin{subfigure}{0.95\textwidth} 
            \includegraphics[width=0.95\textwidth]{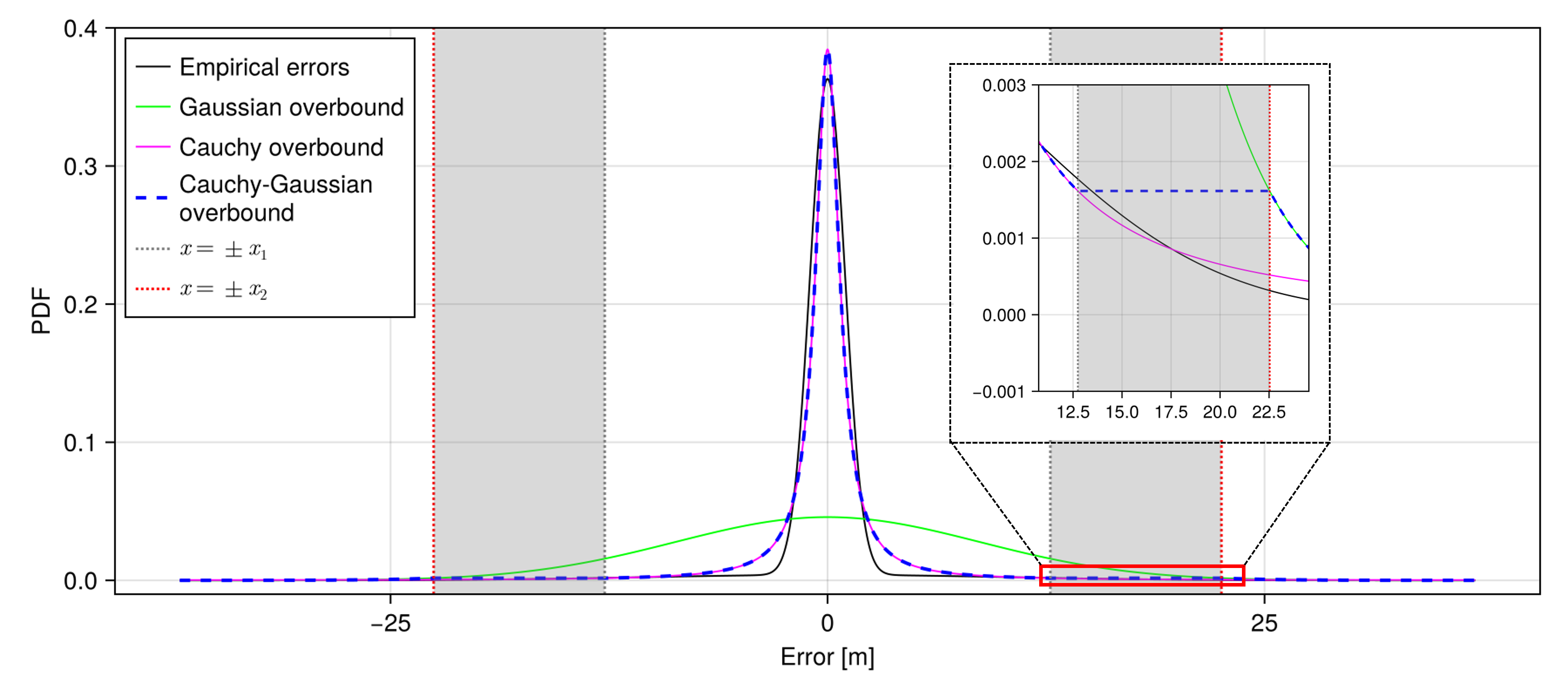}
            \caption{\centering}
            \label{fig: sym ovb pdf}
        \end{subfigure}
        \caption{The Cauchy-Gaussian overbounding results for zero-mean BGMM error defined in Equation \eqref{eq:example_error_type_I} (s.u. profile). Three views are exhibited, including (a) CDF; (b) quantile-scale CDF; (c) PDF. The gray-shaded rectangles denote the two tangential transition regions on the left and right. The Gaussian and Cauchy overbounds are plotted for reference in all the subfigures.}
        \label{fig: demo su}
\end{figure}

The symmetry of the zero-mean Cauchy and Gaussian distributions ensures that PDF of $F_{combined}$ is inherently symmetric, as indicated by Figure \ref{fig: xt pdf}. To further maintain the unimodality of the combined distribution, tangential line segments $T_L(x)$ and $T_R(x)$ are constructed as a transition from the Cauchy core to Gaussian tails. {Appendix \ref{app 1.5: tangent existence} demonstrates that a tangential line always exists between the Cauchy and Gaussian overbounds against heavy-tailed errors.} With the same s.u. empirical error distribution analyzed in Figure \ref{fig: tentative combined overbound}, the overbounding procedures using tangential transition are exhibited in Figure \ref{fig: demo su} . Take the right half $T_R(x)$ for instance, the two tangential points on CDFs are denoted as $\left(x_1, F_C(x_1)\right)$ and $\left(x_2, F_G(x_2)\right)$, as shown in the inset plot in Figure \ref{fig: sym ovb cdf}. Assume that $0<x_1<x_2$, and then the tangential transition should satisfy the following conditions:
\begin{align}
    K=f_C(x_1)=f_G(x_2)=\frac{F_G(x_2)-F_C(x_1)}{x_2-x_1}\ ,
\end{align}
where $K$ is a constant that represents the slope of the tangential segments. The transition on the right half can thus be expressed as:
\begin{equation}
    T_R(x; x_1, x_2)  = K  \cdot \left(x-x_1 \right) + F_C(x_1)\quad \forall x_1 \leq x \leq x_2. \label{equ: transition R}
\end{equation}
Due to the s.u. property of both the Cauchy and Gaussian distributions, the tangential points on the left region are $\left(-x_1, F_C(-x_1)\right)$ and $\left(-x_2, F_G(-x_2)\right)$, and the left tangential transition is given by:
\begin{equation}
    T_L(x; x_1, x_2)  = K  \cdot \left(x+x_1 \right) + F_C(-x_1)\quad \forall -x_2 \leq x \leq -x_1.
    \label{equ: transition L}
\end{equation}
Appendix \ref{app 2: ovb during transition in s.u. cases} proves that the overbounding properties still exists within the transition region.

Figure \ref{fig: demo su} demonstrates the transition in three views, including CDF, quantile-scale\footnote{\enquote{Quantile-scale} indicates that the CDF values are transformed by $\Phi^{-1}(\cdot)$ to the equivalent standard normal quantiles.} CDF, and PDF. Notably, in Figure \ref{fig: sym ovb normcdf}, the Gaussian overbound is processed by the standard normal quantile function $\Phi^{-1}(\cdot)$ and behaves as a straight line. Besides, the left transition region $-x_2 \leq x \leq -x_1$  and the right transition region $x_1 \leq x \leq x_2$ are highlighted by the grey shaded area. To illustrate the details in transitions, an inset plot zooming in the right transition is provided for each of the three subfigures. The left transition is omitted for the present due to the similarity. In the zoom-in CDF view in Figure \ref{fig: sym ovb cdf}, the right transition $T_R(x; x_1, x_2)$ performs as a straight line segment, as defined in Equation \eqref{equ: transition R}. As shown in the inset of Figure \ref{fig: sym ovb normcdf}, the shape of the tangential transition is distorted due to the transformation $\Phi^{-1}(\cdot)$. In Figure \ref{fig: sym ovb pdf}, the tangential transition in the zoom-in PDF view performs as a horizontal line segment, which confirms the unimodality of the combined distribution.

\textbf{Step 3: Synthesize the two models and transitions} \\
The final step synthesizes the results of the previous steps to give the final overbound of the empirical errors. The CDF of the Cauchy-Gaussian overbound for s.u. errors can be explicitly defined as :
\begin{equation}
    F_{ob}(x)= \left\{
        \begin{array}{lr} 
          \min\left(F_C(x;0,\lambda_o), T_L(x; x_1, x_2), F_G(x;0,\sigma_o)\right)        &  \forall  x \leq 0 \\
          \max\left(F_C(x;0,\lambda_o), T_R(x; x_1, x_2),F_G(x;0,\sigma_o)\right)               &  \forall  x > 0
        \end{array}.
    \right.
    \label{equ: F_ob su}
\end{equation}
Figure \ref{fig: demo su} illustrates this piece-wise CDF overbounds with the blue dashed line. In the right region ($x>0$), the finalized overbound starts with the optimal Cauchy model when the absolute error is below $x_1$, followed by the tangential transitions from $x_1$ to $x_2$. $F_{ob}(x)$ then switches to the optimal Gaussian overbound when absolute error exceeds $x_2$. Since the Cauchy-Gaussian combined distribution in Equation \eqref{equ: F_ob su} is s.u., its overbounding properties can be preserved through convolutions \citep{decleene_defining_2000}. 
% {Notably, with the two distributions in Step 1, we chose a piecewise construction over a weighted-sum model similar to the BGMM for two primary reasons. First, our method is deterministic: the final overbound is explicitly calculated from the optimal Gaussian and Cauchy components, unlike GMMs which often rely on heuristic methods \citep{blanch_position_2008,gao_error_2022}. Second, the synthesis of the components is formalized through a tangential transition, rather than a non-standardized weighting scheme, providing a more structured and reproducible overbound. }

The optimal parameters of the Cauchy-Gaussian overbound on the empirical error defined in Equation \eqref{eq:example_error_type_I} are listed in Table \ref{tab: simu su paras}. The Gaussian component of the proposed overbound represents the location and scale parameters given by the optimal Gaussian overbound, while the Cauchy component stands for the counterparts yielded by the optimal Cauchy overbound. {As can be seen, the Cauchy component is characterized by an optimal scale parameter, $\lambda_o$. The resulting Cauchy core is significantly sharper and more concentrated than the core of the optimal Gaussian overbound. This sharpness allows the final Cauchy-Gaussian overbound to adhere more tightly to the empirical errors in the central region, as evidenced in Figure \ref{fig: demo su}.}

% {The proposed overbound leverages the Cauchy distribution to capture the heavy-tailed properties of empirical errors. While these heavy-tailed features can also be characterized by a mixture of simple (e.g., Gaussian) distributions \citep{blanch_position_2008}, the Cauchy distribution offers distinct advantages. First, it provides a more natural model for the sharp core property of many heavy-tailed error distributions. Second, the Cauchy distribution is defined by only two parameters (location and scale), making its optimal form easier to determine than that of a mixture model, which requires the additional complexity of finding optimal weights between its components.}

{The Cauchy-Gaussian overbound leverages both single-CDF Cauchy and single-CDF Gaussian overbounds to capture the distributional properties of heavy-tailed empirical errors. While heavy-tailed error features can also be characterized by a zero-mean BGMM \citep{blanch_position_2008}, the proposed overbound offers distinct advantages regarding parameter determination. Built on the core overbounding approach introduced by \citet{rife_core_2004}, the Cauchy-Gaussian overbound aims to balance between bound sharpness and implementation simplicity. It achieves an overall tight bound by piecewise synthesizing the sharp single-CDF overbounds against the error curve at the core and tail regions. Moreover, the proposed method replaces heuristic guesswork with a rigorous constrained optimization formulation. This eliminates the subjectivity of parameter tuning, providing a strictly deterministic procedure in contrast to the manual trial-and-error required for a BGMM overbound.}
\begin{table}[htb]
 \caption{Overbounding parameters for the simulated s.u. errors defined in Equation \eqref{eq:example_error_type_I}. $K$ represents the slope of the tangential transition, while $x_1$ and $x_2$ are the tangential points when $x>0$, as discussed in Step 2 of Section \ref{sec: method su errors}}
 \label{tab: simu su paras}
 %text alignment: l -left; c - center; r -right
\begin{tblr}{colspec={X[c]X[c]X[c]},
width=\textwidth,
row{even} = {white,font=\small},
row{odd} = {bg=black!10,font=\small},
row{1} = {bg=black!20,font=\bfseries\small},
hline{Z} = {1pt,solid,black!60},
rowsep=3pt
}
Gaussian component $(\mu_o, \sigma_o)$ & Cauchy component $(m_o, \lambda_o)$ & Tangential transition $(K,x_1,x_2)$   \\
(0, 8.72m) & (0, 0.83m) & (2.00$\times$10$^{-3}$m$^{-1}$, 12.76m, 22.55m)
\end{tblr}
\end{table}

\subsection{Procedures of bounding n.s.u. errors}
\label{sec: method nsu errors}
As discussed in Section \ref{sec: paired ovb}, the overbound for the n.s.u profile can be handled with a pair of CDF bounds. It is also illustrated in Section \ref{sec: method su errors} that the combination of Cauchy core and Gaussian tails can more tightly bound the heavy-tailed empirical profiles. Building upon these findings, this section introduces a three-step procedure to extend the Cauchy-Gaussian overbound to n.s.u. error distributions, which are more frequently observed in real-world GNSS measurements. 

In the following discussion, we use randomly generated samples from the following biased BGMM to illustrate the construction process of the proposed method:
\begin{equation}
    f_e(x)=0.9f_G(x;0,1)+0.1f_G(x;1,10), 
    \label{eq:example_error_type_II}
\end{equation}
where the particular location shift of 1m is given to the second Gaussian component so that the resultant distribution is n.s.u. This empirical error distribution with a bias of 0.02m is displayed by black curves in Figure \ref{fig: demo nsu}.

\begin{figure}[t]
        \centering
        \begin{subfigure}{0.49\textwidth} 
            \includegraphics[width=0.95\textwidth]{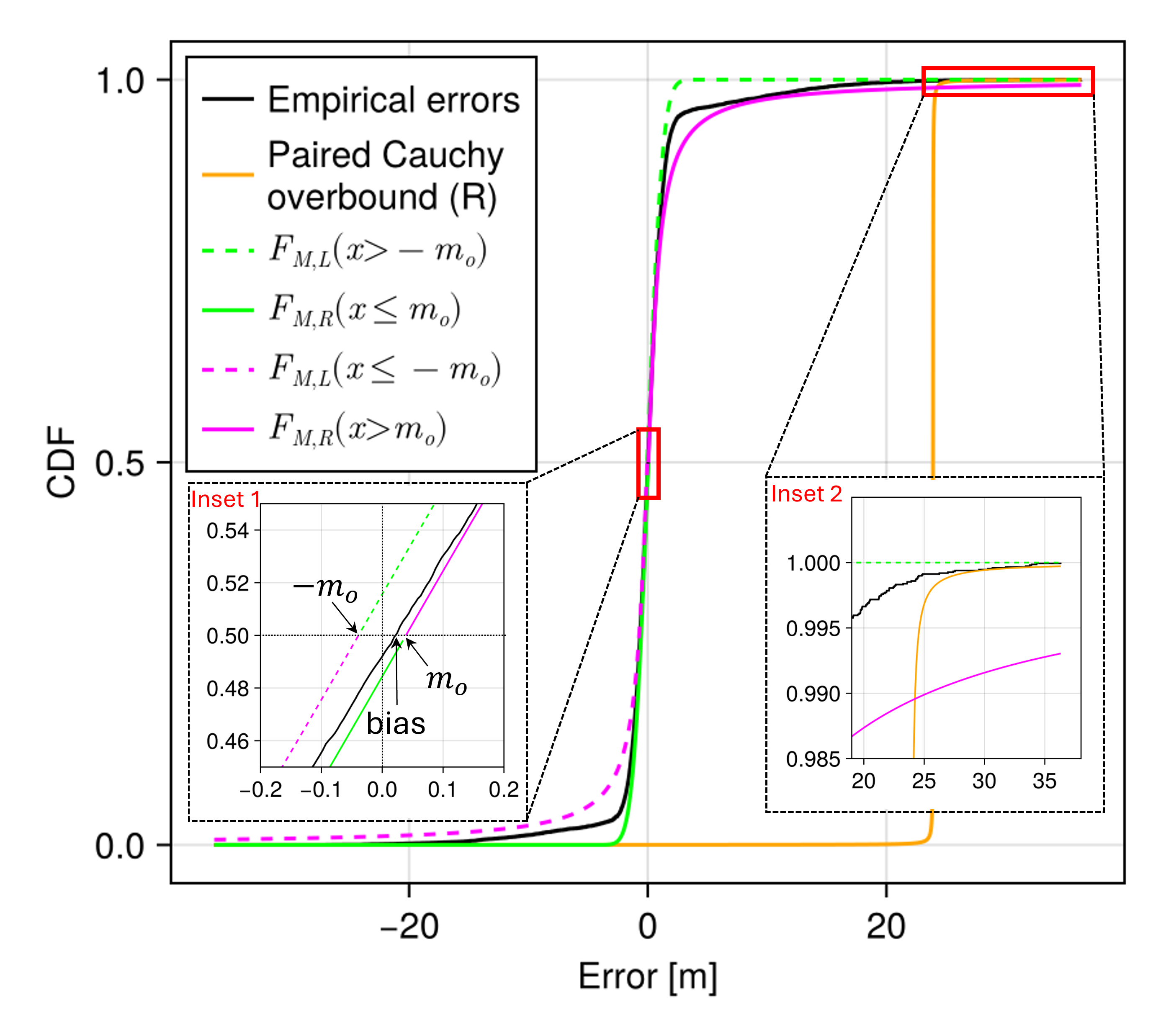}
            \caption{\centering}
            \label{fig: nsu ovb paired cauchy}
        \end{subfigure}
        \hfill
        \begin{subfigure}{0.49\textwidth} % width for the subfigure
            \includegraphics[width=0.95\textwidth]{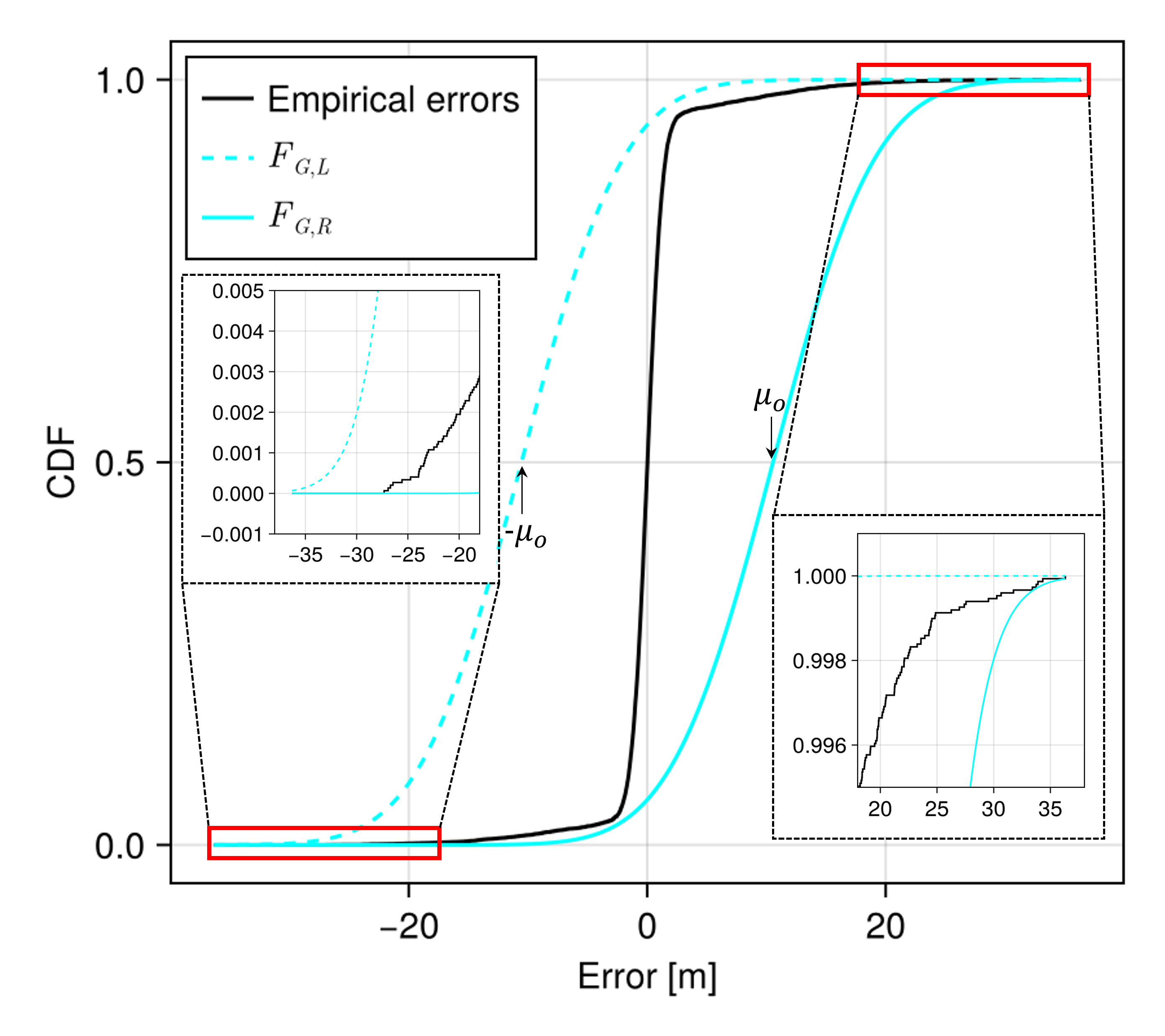}
            \caption{\centering}
            \label{fig: nsu ovb paired gaussian}
        \end{subfigure}
        \\
        \begin{subfigure}{0.49\textwidth} % width for the subfigure
            \includegraphics[width=0.95\textwidth]{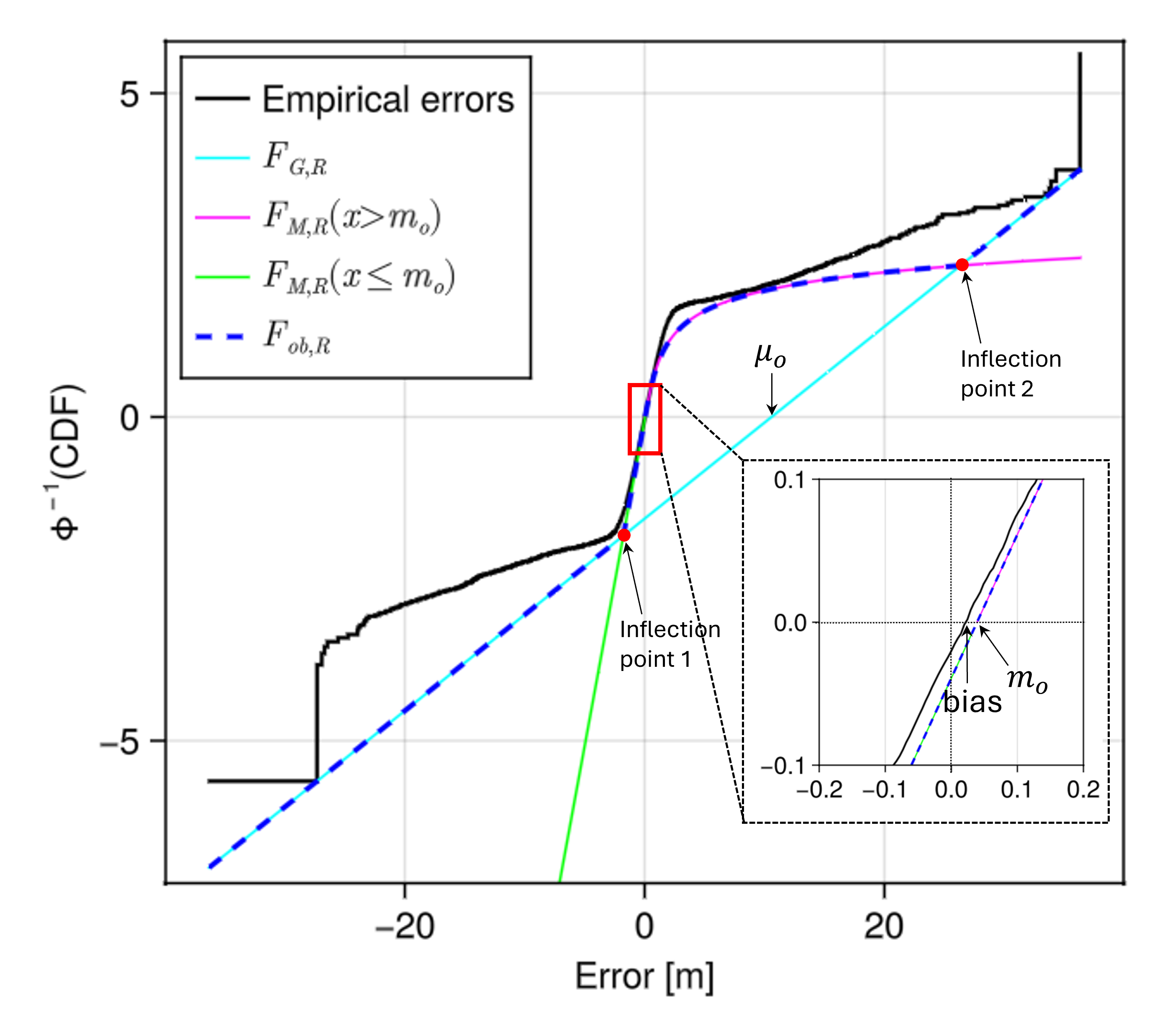}
            \caption{\centering}
            \label{fig: nsu ovb F_ob_R}
        \end{subfigure}
        \hfill 
        \begin{subfigure}{0.49\textwidth} 
            \includegraphics[width=0.95\textwidth]{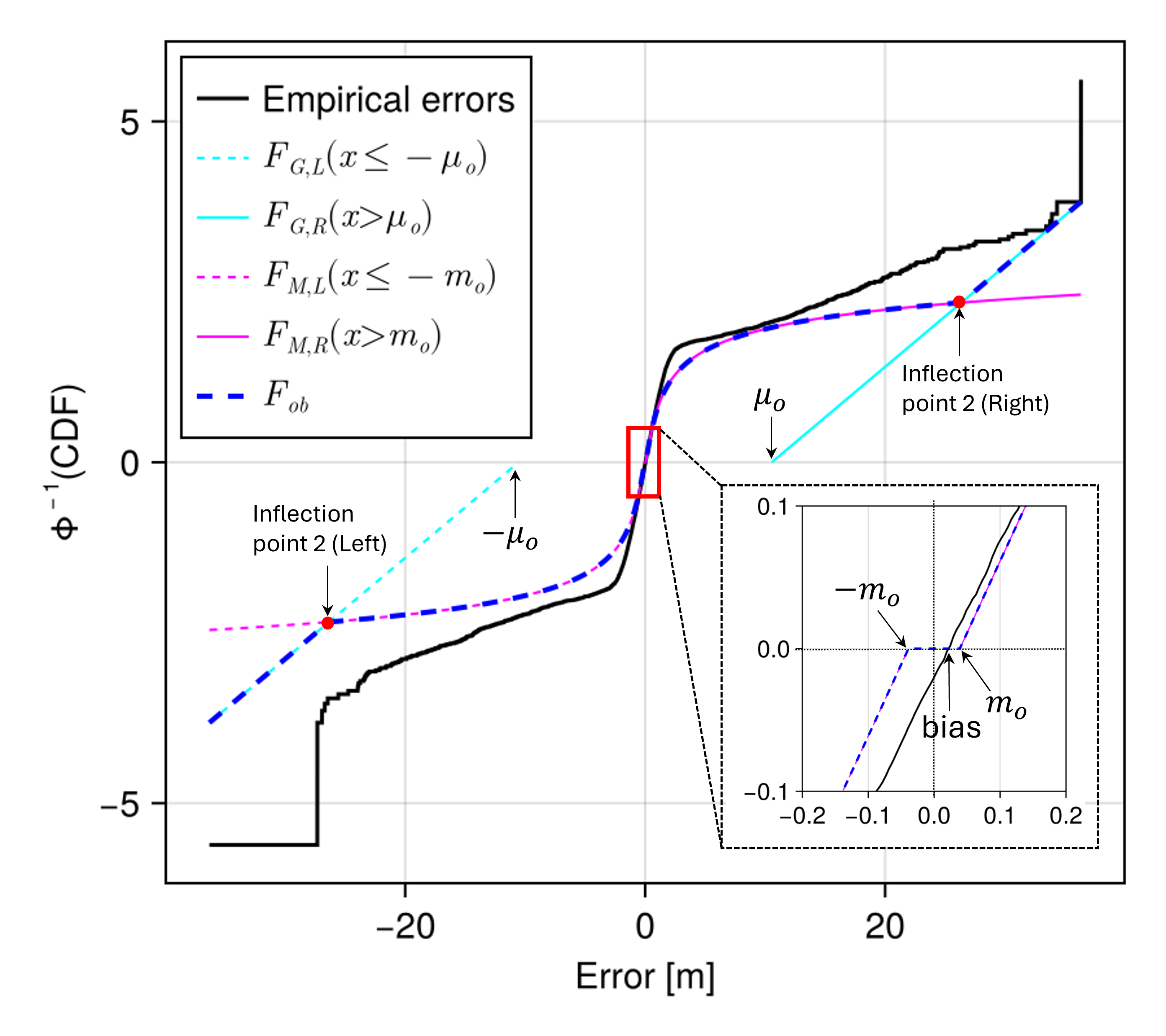}
            \caption{\centering}
            \label{fig: nsu ovb F_ob}
        \end{subfigure}
        \caption{The Cauchy-Gaussian overbounding results for biased BGMM errors defined in Equation \eqref{eq:example_error_type_II} (n.s.u. profile): (a) CDF of the optimized paired CGCM overbound in Step 1; (b) CDF of the optimized paired Gaussian overbound in Step 2; quantile-scale (c) right bound (defined in Equation \eqref{equ: synthesized CG R}) and (d) analog single-CDF (defined in Equation \eqref{equ: F_ob nsu single-cdf analog}) of finalized Cauchy-Gaussian overbound.}
        \label{fig: demo nsu}
\end{figure}

\textbf{Step 1: Construct optimal paired CGCM overbound} \\
\label{sec: paired CGCM}
The first step features finding the paired overbound with the Cauchy distribution. Tentatively, we assume a pair of pure Cauchy distributions to construct the two-sided overbound. According to Equation \eqref{equ: paired L} and \eqref{equ: paired R}, the paired Cauchy overbound should satisfy: 
\begin{align}
    F_C(x; m_L,\lambda_L)  &\geq F_e(x)\quad\forall x, \label{equ: paired Cauchy L} \\
    F_C(x; m_R,\lambda_R)  &\leq F_e(x)\quad\forall x. \label{equ: paired Cauchy R}
\end{align}
However, a typical pair of Cauchy overbounds satisfying the above conditions is found to yield significantly large location parameters ($|m_L| = |m_R|=$ 23.90) yet considerably small scale parameters ($\lambda_L = \lambda_R=$ 0.01). Figure \ref{fig: nsu ovb paired cauchy} visualizes the right half of paired Cauchy overbound, and the left half is inferrable due to symmetry. As can be seen, the paired Cauchy overbound evidently separates from the empirical errors and results in overly conservative bounding. This phenomenon is mainly caused by the extreme heavy-tailedness of the Cauchy distribution. Take the error bounding on the right half, for instance, the long and fat left tail of the Cauchy distribution hinders the Cauchy CDF from completely lying below the empirical distribution (as required in Equation \eqref{equ: paired Cauchy R}) with a small location shift. As a result, the paired Cauchy overbound has to compromise with the constraint in Equation \eqref{equ: paired Cauchy R} using a significantly large location parameter, which induces excessive conservatism.

Inspired by the above analysis, we replace the left tail of the right Cauchy bound and the right tail of the left Cauchy bound with Gaussian tails to prevent the location parameter from over-inflation. Specifically, we propose the paired Cauchy-Gaussian combined models (CGCM) with its left ($F_{M,L}$) and right ($F_{M,R}$) halves constructed as follows:
\begin{align}
    F_{M,L}(x)&= \left\{
        \begin{array}{lr} 
          F_C\left(x;-m,\lambda\right)        &  \forall x \leq -m \\
          F_{G}\left(x;-m,k\lambda\right) &  \forall x > -m
        \end{array},
    \right.  \\
    F_{M,R}(x)&= \left\{
        \begin{array}{lr} 
          F_{G}\left(x;m,k\lambda\right)        &  \forall x \leq m \\
          F_C\left(x;m,\lambda\right) &  \forall x > m
        \end{array}.
    \right.
\end{align}
Notably, we use $k=\sqrt{\frac{\pi}{2}}$ for the paired CGCM. This setting means the centrally-aligned $F_C$ and $F_{G}$ share the same value at $x=-m$ and $x=m$, and their derivatives (i.e., $f_C$ and $f_G$) are equal at these two locations. Therefore, each half (i.e., $F_{M,L}$ or $F_{M,R}$) is continuous and differentiable throughout the coordinate. Moreover, the setting $k=\sqrt{\frac{\pi}{2}}$ establishes a one-to-one relationship between the aligned Cauchy and Gaussian distributions, which avoids introducing additional parameters.\par

Notably, the right half ($F_{M,R}$) of the paired CGCM overbound contains Gaussian when $x\leq m$ and Cauchy when $x>m$. Due to the Cauchy distribution's heavy-tailedness, the probability mass of the combined distribution is transferred rightwards, making the mean exceed the median (i.e., right-skewed). Compared to a Gaussian tail, the Cauchy tail potentially enables the right-half CGCM to bound right-biased heavy-tailed empirical errors with a lower location parameter. Similarly, the left-skewed $F_{M,L}$ can largely reduce the location shift when bounding left-biased heavy-tailed error distributions.

The following optimization problem is constructed to solve the optimal location ($m_o$) and scale ($\lambda_o$) parameters of the paired CGCM overbound:
\begin{subequations}
\begin{align}
        \quad& m_o, \lambda_o = \argmin_{m, \lambda} \sum \limits_{x\in \Omega} ||F_{M,L}(x)-F_e(x)|| + ||F_{M,R}(x)-F_e(x)|| \\
        s.t.\quad \quad & F_{M,L}(x) \geq F_e(x)\quad \forall x \in \Omega \\
         & F_{M,R}(x) \leq F_e(x)\quad \forall x \in \Omega,
    \end{align}
	\label{opt1: getting optimal combined pairs}
\end{subequations}
where $||\cdot||$ denotes the Euclidean norm and $\Omega$ is the range of empirical errors. The optimization problem minimizes the discrepancy between the empirical and overbounding distributions, thereby achieving tight bounding\par

Figure \ref{fig: nsu ovb paired cauchy} displays the optimal left and right bounds using CGCM. Each bound contains the Cauchy (magenta) and Gaussian (green) components. The empirical errors are bounded with a substantially small location parameter ($m_o=$ 0.04m). $F_{M,R}$ transits from Gaussian to Cauchy at $x=m_o$, while $F_{M,L}$ transits from Cauchy to Gaussian at $x=-m_o$. The inset plot 1 highlights the smooth junction between the Gaussian and Cauchy distributions in the right (at $x=m_o$) and left regions (at $x=-m_o$). Take the right region, for instance, Figure \ref{fig: nsu ovb paired cauchy}  also validates the capability of CGCM to generate much tighter bounds than the paired Cauchy approach. \par

Notably, the right tail of $F_{M,R}$ and left tail of $F_{M,L}$ employ the Cauchy distribution, which is proven to be significantly more heavy-tailed than the Gaussian distribution in Section \ref{sec: cauchy distribution}. This phenomenon is also reflected in the inset plot 2 in Figure \ref{fig: nsu ovb paired cauchy}, where the heavy-tailed $F_{M,R}$ bounds the empirical errors insufficiently tightly at the right tail region. The heavy-tailed characteristic of the Cauchy tail will be further enhanced after convoluting multiple ranging error sources, eventually producing excessively conservative PLs in the position domain. To reduce the conservatism, we propose replacing these Cauchy tails with Gaussian tails from the optimal paired Gaussian overbound, which will be illustrated in the following section.

\textbf{Step 2: Construct optimal paired Gaussian overbound} \\
The second step aims to determine the optimal pair of Gaussian overbounds with the location parameter $\mu_o$ and scale parameter $\sigma_o$. The optimal left ($F_{G,L}(x)$) and right ($F_{G,R}(x)$) overbounds are the ones that can give the minimum least square sum of CDF differences with the empirical CDF $F_e(x)$. Similar to Step 1, the optimization problem can be expressed as:
\begin{subequations}
\begin{align}
        \quad& \mu_o, \sigma_o = \argmin_{\mu, \sigma} \sum \limits_{x\in \Omega} ||F_{G,L}(x)-F_e(x)|| + ||F_{G,R}(x)-F_e(x)|| \\
s.t.\quad\quad& F_{G,L}(x) = F_{G}(x;-\mu,\sigma) \geq F_e(x)\quad \forall x \in \Omega \\
& F_{G,R}(x)=F_{G}(x;\mu,\sigma) \leq F_e(x)\quad \forall x \in \Omega.
    \end{align} \label{opt2: getting optimal gaussian pairs}
\end{subequations}

The CDF of the optimal paired Gaussian overbounds are shown as cyan lines in Figure \ref{fig: nsu ovb paired gaussian}. The zoom-in views confirm the strict left and right bounds are satisfied at tail regions according to Equation \eqref{equ: paired L} and \eqref{equ: paired R}.

\textbf{Step 3: Synthesize the optimal paired CGCM overbound and paired Gaussian overbound} \\   
The final step integrates optimal paired overbounding models from the first two steps. Specifically, the finalized Cauchy-Gaussian overbound for n.s.u. error distributions are defined as follows:
\begin{align}
    F_{ob,L}(x) &=  \min\left(F_{M,L}(x), F_{G,L}(x)\right) \quad\forall x \in \Omega, \label{equ: synthesized CG L}\\
      F_{ob,R}(x) &=  \max\left(F_{M,R}(x), F_{G,R}(x)\right) \quad\forall x \in \Omega.  \label{equ: synthesized CG R}
\end{align}
Mathematically, $F_{ob,L}(x)$ represents the point-wise infimum of $F_{M,L}(x)$ and $ F_{G,L}(x)$, while $F_{ob,R}(x)$ is the point-wise supremum of $F_{M,R}(x)$ and $ F_{G,R}(x)$. Different from Equation \eqref{equ: F_ob su} for bounding s.u. errors, the tangential transitions are unnecessary in bounding n.s.u. error distributions. The reasons are that the paired bounding approach does not require the left and right bounds to be zero-mean s.u. distributions to maintain the overbounding properties through convolutions \citep{rife_paired_2006}. \par

Take the right bound, for instance, Figure \ref{fig: nsu ovb F_ob_R} shows how $F_{ob,R}(x)$ is constructed from the right half of optimal paired CGCM overbound ($F_{M,R}(x)$) and optimal paired Gaussian overbound ($F_{G,R}(x)$). Notably, Gaussian distributed curves $F_{G,R}$ and $F_{M,R}(x\leq m_o)$ in the quantile-scale CDF  transformed by $\Phi^{-1}(\cdot)$ behave as straight lines. According to Equation \eqref{equ: synthesized CG R}, the piecewise right bound begins with $F_{G,R}(x)$ at the far end of negative error, transits to $F_{M,R}(x)$ after the inflection point 1, and switches back to $F_{G,R}(x)$ after inflection point 2, till the positive far end. The left bound $F_{ob,L}(x)$ is omitted in Figure  \ref{fig: nsu ovb F_ob_R} but can be inferred due to the symmetry with the right bound. Correspondingly, the following analog single-CDF overbound $F_{ob}(x)$ based on the paired overbound Equation \eqref{equ: synthesized CG L} and \eqref{equ: synthesized CG R} can be constructed through Equation \eqref{equ: single-cdf analog}:
\begin{equation}
F_{ob}(x)= \left\{
    \begin{array}{lr} 
      \min\left(F_{M,L}(x), F_{G,L}(x)\right) & \forall  F_{ob,L} < \frac{1}{2} \\
      \frac{1}{2}      &  \text{otherwise}\\
      \max\left(F_{M,R}(x), F_{G,R}(x)\right)  &  \forall  F_{ob,R} > \frac{1}{2}
    \end{array}.
\right.
\label{equ: F_ob nsu single-cdf analog} 
\end{equation}
Figure \ref{fig: nsu ovb F_ob} depicts the quantile-scale CDF of $F_{ob}$ and displays its relationship with the paired CGCM overbound ($F_{M,L},F_{M,R}$) and paired Gaussian overbound ($F_{G,L},F_{G,R}$). As can be seen, $F_{ob}$ employs the paired CGCM overbound at small absolute error values. Subsequently, $F_{ob}$ transits to the optimal paired Gaussian overbound at the tail region after the left and right intersection point 2. Finally, the inset plots in both Figure \ref{fig: nsu ovb F_ob_R} and \ref{fig: nsu ovb F_ob} show that the empirical error bias is bounded between $-m_o$ and $m_o$ at the core region. \par
\begin{table}[htb]
 \caption{Cauchy-Gaussian overbounding parameters for the simulated n.s.u. errors defined in Equation \eqref{eq:example_error_type_II}.}
 \label{tab: simu nsu paras}
 %text alignment: l -left; c - center; r -right
\begin{tblr}{colspec={X[c]X[c]},
width=\textwidth,
row{even} = {white,font=\small},
row{odd} = {bg=black!10,font=\small},
row{1} = {bg=black!20,font=\bfseries\small},
hline{Z} = {1pt,solid,black!60},
rowsep=3pt
}
Optimal paired CGCM overbound $(m_o, \lambda_o)$ & Optimal paired Gaussian overbound $(\mu_o, \sigma_o)$   \\
(0.04m, 0.79m) & (10.55m, 6.75m)   
\end{tblr}
\end{table}
Table \ref{tab: simu nsu paras} displays the parameters of the Cauchy-Gaussian overbound $F_{ob}$ for the simulated n.s.u. errors defined in Equation \eqref{eq:example_error_type_II}, which contains the optimal location and scale parameters from the paired CGCM overbound ($m_o$, $\lambda_o$) and the paired Gaussian overbound ($\mu_o$, $\sigma_o$). The CGCM component features an extremely small location parameter ($m_o=$ 0.04m), which adheres the bounding distribution closely to the empirical curves at the central region. Although the Gaussian component has a large location parameter ($\mu_o=$ 10.55m), it only contributes to the bounding at the tail regions, as reflected in Figure \ref{fig: nsu ovb F_ob_R} and \ref{fig: nsu ovb F_ob}. Therefore, when bounding central-region errors, $\mu_o$ in the final Cauchy-Gaussian overbound does not introduce excessive conservatism compared to the corresponding parameter in the paired Gaussian overbound.

\subsection{Position domain bounding}
{For GNSS positioning, the measurement errors are projected to the position domain. It is essential to derive the error bound for the positioning error. Consider $N$ ranging error sources denoted by the vector $\bm{\varepsilon}$, where each element $\bm{\varepsilon_i}$ has an index $i\in\{1,2,...,N\}$. {In this paper, we assume error sources in $\bm{\varepsilon}$ are mutually-independent.} The vertical positioning error (VPE) can be expressed as:
\begin{equation}
    \text{VPE} = \sum\limits^{N}_{i=1} \textbf{S}_{3,i}\bm{\varepsilon}_i, \label{equ: VPE}
\end{equation}
where $\textbf{S}_{3,i}$ is the third-row and $i^{th}$-column element in the projection matrix $\mathbf{S}$ of the least-square or weighted least-square solution. $\text{VPE}$ is mathematically the linear combination of the $N$ error sources, each of which has an individual error distribution. As such, the PDF of $\text{VPE}$, $f_\text{VPE}$, is a joint distribution given by multiple measurement error distributions. According to \citet{yan_principal_2024}, $f_\text{VPE}$ can be formulated as:
\begin{equation}
    f_{\text{VPE}}(x)= \left( \prod_{i=1}^{N} \frac{1}{|\textbf{S}_{3,i}|}\right)   \left(\Conv_{i=1}^{N}  f_{\bm{\varepsilon}_i}\left(\frac{x}{|\textbf{S}_{3,i}|} \right) \right), \label{equ: f_VPE 1}
\end{equation}
where $f_{\bm{\varepsilon}_i}$ denotes the empirical distribution related to the $i^{th}$ measurement error source. The operators $\Pi$ and $\circledast$ represent the multiplication and convolution of multiple terms, respectively. A proof of Equation \eqref{equ: f_VPE 1} is provided in Appendix \ref{app 3: proof of f_VPE 1}.  For the convenience of computations,  the worst-case error source ${\bm{\varepsilon}_w}$ (where $w$ is a certain index among $\{1,2,\ldots,N\}$) with the largest variance will be selected to represent each error source $\bm{\varepsilon}_i$. In this way, the most conservative empirical distribution $\overline{f}_{\text{VPE}}$ is constructed by:
\begin{equation}
     \overline{f}_{\text{VPE}}(x)= \left( \prod_{i=1}^{N} \frac{1}{|\textbf{S}_{3,i}|}\right)   \left(\Conv_{i=1}^{N}  f_{\bm{\varepsilon}_w}\left(\frac{x}{|\textbf{S}_{3,i}|} \right) \right).
\end{equation}
As such, to obtain the overbound for the most conservative VPE, only the distribution $f_{\bm{\varepsilon}_{w,ob}}$ associated with the overbound against the measurement error $\bm{\varepsilon}_w$ needs to be determined. Notably, $f_{\bm{\varepsilon}_{w,ob}}$ denotes the PDF of the bound for s.u. error. For the n.s.u. error, it represents the PDF of either the right or left half of the paired bound. The position error overbounding distribution can be expressed by:
\begin{equation}
     \overline{f}_{\text{VPE},ob}(x)= \left( \prod_{i=1}^{N} \frac{1}{|\textbf{S}_{3,i}|}\right)   \left(\Conv_{i=1}^{N}  f_{\bm{\varepsilon}_{w,ob}}\left(\frac{x}{|\textbf{S}_{3,i}|} \right) \right). \label{equ: position ob continuous}
\end{equation}
A discretized form of position-domain overbound, $\overline{X}_{\text{VPE},ob}$, is constructed as follows \citep{yan_principal_2024}:
\begin{equation}
     \overline{X}_{\text{VPE},ob}[k]= \left( \prod_{i=1}^{N} \frac{1}{|\textbf{S}_{3,i}|}\right)   \left(\Conv_{i=1}^{N}  X_{\bm{\varepsilon}_{w,ob}} \right),
\end{equation}
where $X_{\varepsilon_{w,ob}}$ is the discretized $f_{\bm{\varepsilon}_{w,ob}}$. As for $\overline{X}_{\text{VPE},ob}[k]$ ($k\in\{1,2,\ldots, M\}$), $k$ is the index and $M$ is the length of the discrete VPE overbounding distribution. Accordingly, the discretized range of position error values is noted as $T[k], k\in\{1,2,\cdots,M\}$, with a unit interval length $\Delta t$. It is proven that this discretization method can preserve the overbounding properties through convolution. Although the discretized position error bound is calculated numerically, the convolution's computation time can be significantly reduced by using the fast Fourier transform for practical considerations \citep{nussbaumer_fast_1982}. For details, one can refer to \citep{yan_principal_2024}. For a paired overbounding approach, the position error bound is typically computed by convolving both the left and right bounds. However, since the paired Cauchy-Gaussian overbound is symmetric, the left bound can be inferred from the right. Therefore, for the remainder of this discussion, we will only utilize the right bound to generate the position-domain bound $\overline{X}_{\text{VPE},ob}[k]$.}\par
{Under the requirement of integrity risk $P_{HMI}$, a particular index $k_p\in  \{1,2,...M\}$ is determined such that \citep{yan_principal_2024}:
\begin{align}
    \sum \limits_{k=1}^{k_p-1} \overline{X}_{\text{VPE},ob}[k] \cdot \Delta t  &\leq 1-\frac{P_{HMI}}{2} , \\
    \sum \limits_{k=1}^{k_p} \overline{X}_{\text{VPE},ob}[k] \cdot \Delta t &>  1-\frac{ P_{HMI}}{2},
\end{align}
and the VPL is obtained by
\begin{equation}
    \text{VPL} = T[k_p]. \label{equ: VPL}
\end{equation}}

\section{Numerical Experiments}
\label{sec: experiments and evaluations}
In this study, we compare the bounding performance of the proposed method with the single-CDF Gaussian overbound and the two-step Gaussian overbound. {The empirical error distributions are constructed from double-differenced pseudorange measurements generated by a differential GNSS (DGNSS) model \citep{misra_global_2006, parkinson_basis_1987}.} {Inspired by the work of \citet{larson_gaussianpareto_2019}, this paper assumes that the DGNSS measurement errors are mutually independent, {as the marginal inter-measurement correlation is expected to have a negligible impact on the computed VPL.}} For s.u. error distributions, the experiment will be conducted on a simulated DGNSS error dataset because the strictly s.u. property can hardly exist in real-world measurement errors. For n.s.u. errors, the bounding performance is validated through a real DGNSS dataset collected in the Hong Kong urban environment.

\subsection{Bounding performance for s.u. errors}
\label{sec: exp su errors}
In Section \ref{sec: method su errors}, we use an example zero-mean BGMM to illustrate the bounding performance of the proposed method on a s.u. heavy-tailed error distribution in the range domain. {In this section, we further use this BGMM model to represent the DGNSS measurement error distribution for the entire time series of each satellite and examine the bounding performance in the position domain.} Specifically, we use the open-sourced MATLAB Algorithm Availability Simulation Tool (MAAST)  \citep{jan_matlab_2001} to simulate the positions of GPS and Galileo satellites every 100 seconds over 24 hours. We set two imaginary locations (displayed in Appendix \ref{app 0: LLH, MAAST}) for the receiver and reference station, separated by a distance of approximately 5.58km. {For each epoch, we generate random samples using a pseudorandom number generator in MATLAB from the zero-mean BGMM in Equation \eqref{eq:example_error_type_I} and add them to the true differential ranges to create the simulated DGNSS measurements.}
\begin{figure}[t]
    \centering
    \includegraphics[width = 0.95\textwidth]{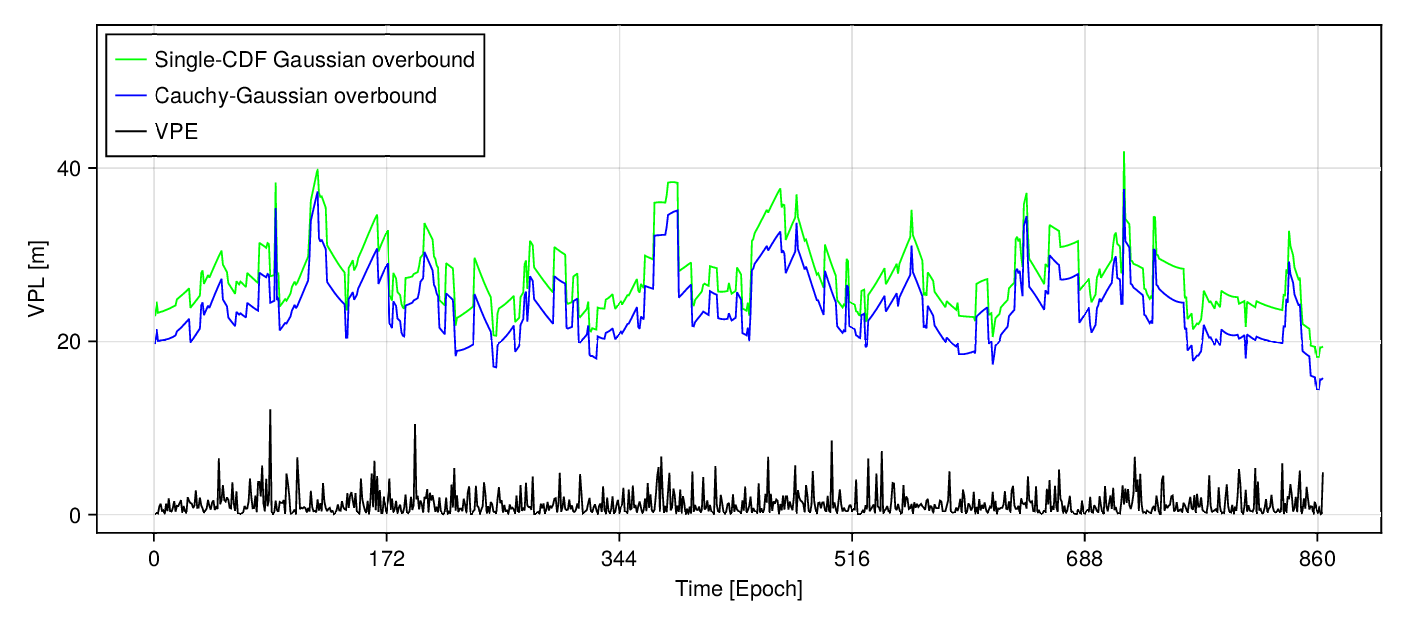}
    \caption{VPL of single-CDF Gaussian and Cauchy-Gaussian overbounding methods, with $P_{HMI} = 1\times 10^{-9}$. The VPEs are also plotted for reference.}
    \label{fig: simu su pl} 
\end{figure}
Over the 864 epochs spanning 24 hours, the positioning solution based on DGNSS measurements is calculated using the least-squares method and the VPEs are exhibited in Figure \ref{fig: simu su pl}. The mean VPE is 1.23m and the maximum reaches 12.15m. In each epoch, we calculate the VPL of the single-CDF Gaussian and Cauchy-Gaussian overbounding methods via Equation \eqref{equ: VPL}, and plot the results in Figure \ref{fig: simu su pl}. The calculation of VPL adopts the same discretization method and settings in \citep{yan_principal_2024}, where the unit interval length $\Delta t$ is 0.01m and $P_{HMI}$ is set to be 1$\times$10$^{-9}$. As shown in Figure \ref{fig: simu su pl}, the VPL produced by the Cauchy-Gaussian overbound is smaller than that by the single-CDF Gaussian overbound in all epochs. Table \ref{tab: simu su pl improve} summarizes the percentage of reduction in VPL achieved by Cauchy-Gaussian overbound compared to the Gaussian overbound. Specifically, the proposed method reduces the VPL by 14.95\% on average. The maximum reduction reaches 26.14\% and the minimum reduction still has 6.78\%. {Besides, the 25th (Q1), 50th (Q2), and 75th (Q3) percentiles varying from 12.61\% to 16.81\% are provided for a more detailed overview of VPL reduction.}

\begin{table}[htb]
 \caption{{Reduction in VPL by the Cauchy-Gaussian overbound compared to the single-CDF Gaussian overbound in the simulated s.u. error dataset. Q1, Q2, and Q3 represent the 25th, 50th and 75th percentiles, respectively.}}
 \label{tab: simu su pl improve}
 %text alignment: l -left; c - center; r -right
{
\begin{tblr}{colspec={X[c]X[c]X[c]X[c]X[c]X[c]X[c]},
width=\textwidth,
row{even} = {white,font=\small},
row{odd} = {bg=black!10,font=\small},
row{1} = {bg=black!20,font=\bfseries\small},
hline{Z} = {1pt,solid,black!60},
rowsep=3pt
}
Metric & Average & Maximum  & Minimum & Q1 & Q2 & Q3\\
Percentage of reduction & 14.95\% & 26.14\% & 6.78\%  & 12.62\% & 14.31\% & 16.81\% 
\end{tblr}}
\end{table}

\subsection{Bounding performance for n.s.u. errors}
\label{sec: exp nsu errors}
DGNSS measurements collected in the Hong Kong urban environment are used to verify the feasibility of the proposed method on bounding real-world n.s.u. errors. The receiver and the reference station are separated by about 4.74km, with the locations displayed in Appendix \ref{app 0: LLH, MAAST}. {This urban dataset collects L1 GPS, BeiDou, and GLONASS signals using u-blox ZED-F9P at a frequency of 1 Hz on June 28$^{th}$, 2024, producing DGNSS measurements covering one hour. Following the settings in \citep{yan_principal_2024}, we select fault-free DGNSS measurements by filtering for observations with an elevation angle above 30$^\circ$ and a signal-to-noise ratio (SNR) greater than 35 dB. To ensure the dataset exhibits the desired heavy-tailed properties for the bounding experiment, we select double-differenced measurement errors from satellites at low elevation angles (30$^{\circ}-$35$^{\circ}$), where satellite signals are known to be susceptible to multipath effects in urban scenarios \citep{peretic_statistical_2025}.} {After the filtering, the 1-hour dataset generates 2887 effective epochs, where 4562 DGNSS error samples constitute the empirical distribution. The number of unique contributing satellites across the epochs are displayed in Figure \ref{fig: unique sats}. } In the following experiments, the two-step Gaussian overbound \citep{blanch_gaussian_2018} and the non-Gaussian NavDEN overbound \citep{rife_overbounding_2012}  are utilized as the benchmark to compare with the proposed Cauchy-Gaussian overbound. Details on the NavDEN model are provided in Appendix \ref{app 4: navden model}.

\begin{figure}[t]
    \centering
    \includegraphics[width = 0.95\textwidth]{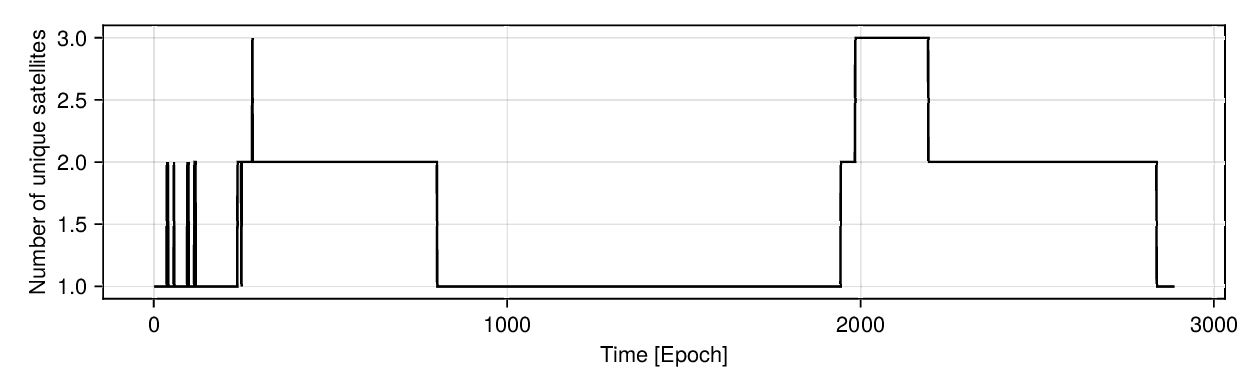}
    \caption{{The number of unique contributing satellites in each epoch from the filtered dataset.}}
    \label{fig: unique sats} 
\end{figure}

\begin{figure}[t]
        \centering
        \begin{subfigure}{0.49\textwidth} % width for the subfigure
            \includegraphics[width=0.95\textwidth]{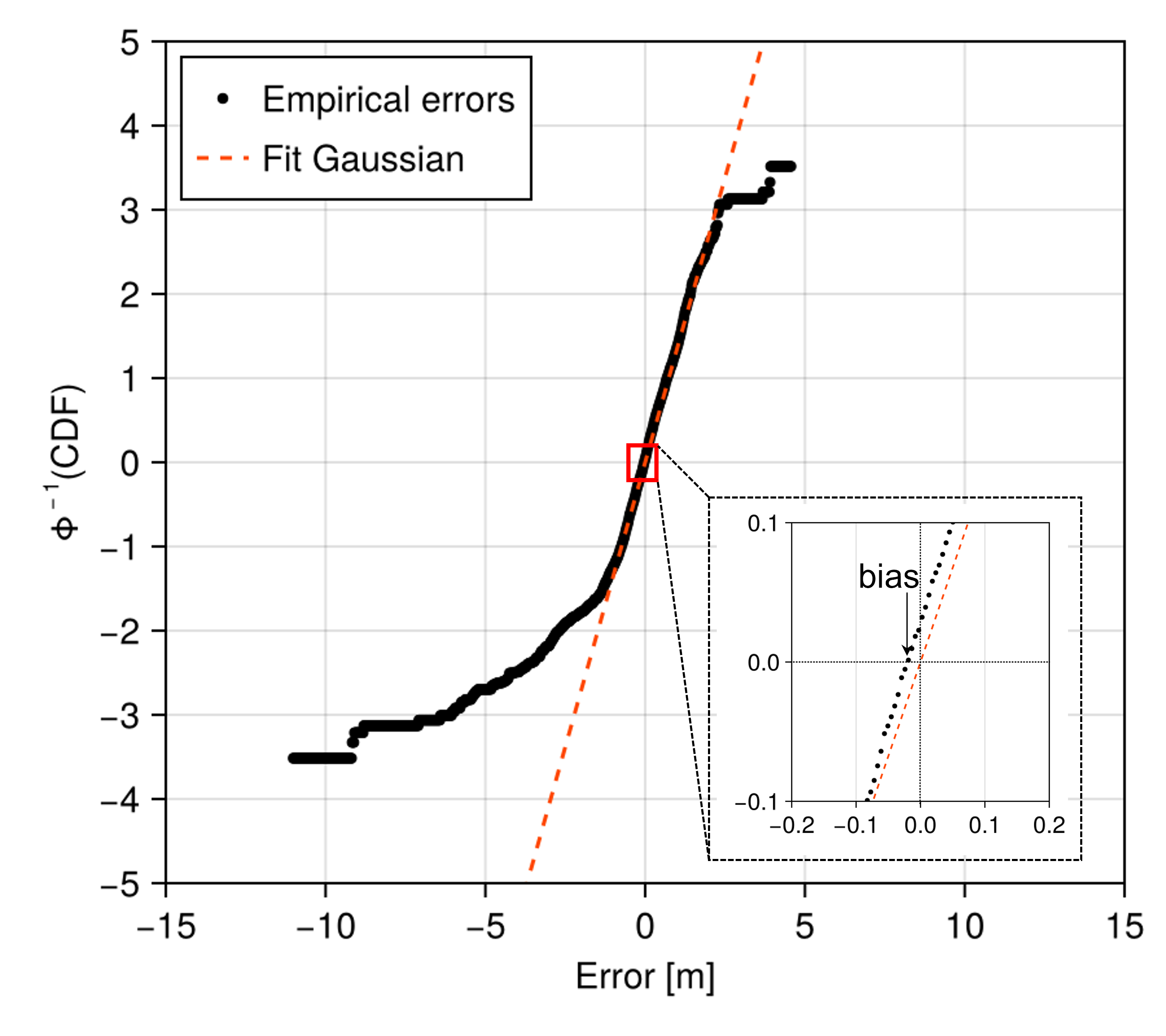}
            \caption{\centering}
            % \label{fig: real nsu2 qq error}
        \end{subfigure}
        \hfill 
        \begin{subfigure}{0.49\textwidth} 
            \includegraphics[width=0.95\textwidth]{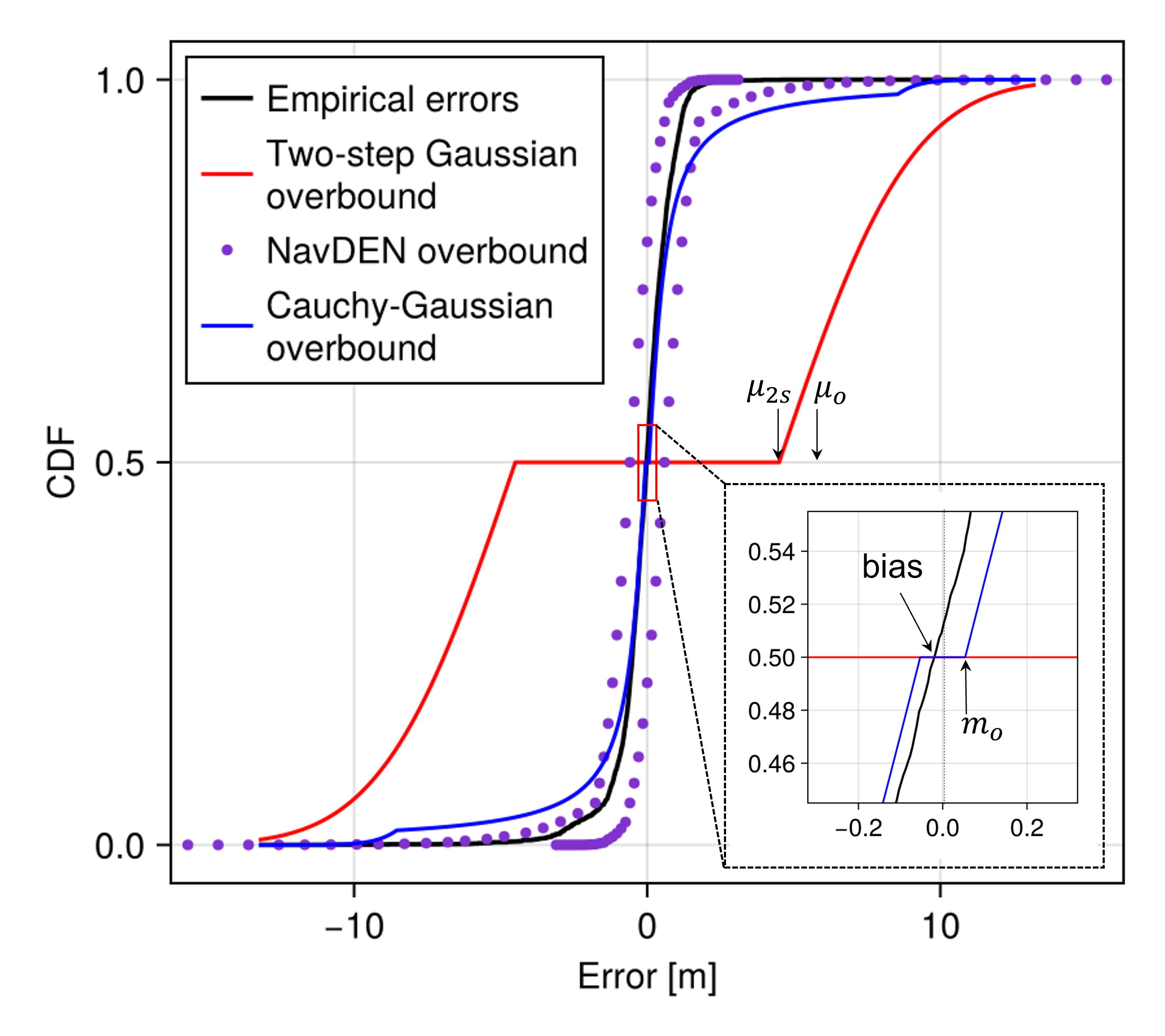}
            \caption{\centering}
            % \label{fig: real nsu2 ovb wtnormal}
        \end{subfigure}
        \\
        \begin{subfigure}{0.49\textwidth} % width for the subfigure
            \includegraphics[width=0.95\textwidth]{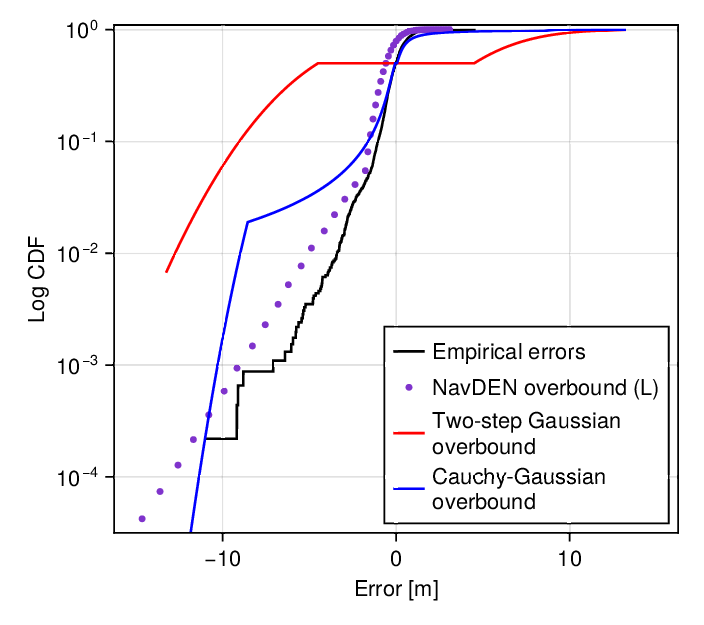}
            \caption{\centering}
            % \label{fig: real nsu2 ovb 1}
        \end{subfigure}
        \hfill 
        \begin{subfigure}{0.49\textwidth} 
            \includegraphics[width=0.95\textwidth]{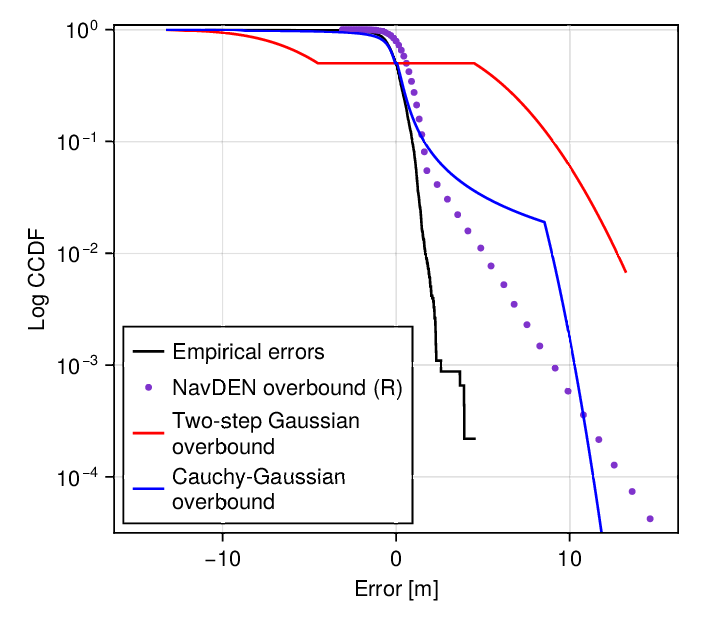}
            \caption{\centering}
            % \label{fig: real nsu2 ovb 2}
        \end{subfigure}
        \caption{{Evaluation of overbounding performance for n.s.u. errors in the urban dataset. (a) Quantile-scale CDF showing empirical DGNSS error during one hour with observed elevation angle between 30$-$35$^\circ$. Overbounding results of two-step Gaussian and Cauchy-Gaussian in three views: (b) CDF; (c) Log CDF; (d) Log CCDF, where the Cauchy-Gaussian overbound takes the analog single-CDF form defined in Equation \eqref{equ: F_ob nsu single-cdf analog}.}}
        \label{fig: real nsu2 ovb}
\end{figure}

{The quantile-scale CDF in Figure \ref{fig: real nsu2 ovb}a confirms the error distribution is neither symmetric nor unimodal, with a bias of -0.019m revealed by the inset plot. For comparison, a zero-mean Gaussian model (red dashed line) is fitted to the central 95\% of the empirical data. While the empirical curve (black) aligns well with this fit in the core region, it deviates significantly in the tails. This separation begins around -1.5m on the left and 2m on the right, with the larger deviation in the left tail indicating it is fatter and longer than the right tail.}

{Table \ref{tab: real nsu2 paras} summarizes the parameters of the NavDEN, the two-step Gaussian, and Cauchy-Gaussian overbounding methods. Multiple parameters are employed to define different parts of a NavDEN overbound and their meanings have been summarized in Appendix \ref{app 4: navden model}. The location and scale parameters (i.e., ($\mu_{2s}$, $\sigma_{2s}$)) of the two-step Gaussian overbound are generated by the open-source MATLAB tool in \citep{blanch_gaussian_2018}.} The optimal parameters of the Cauchy-Gaussian overbound defined in Equation \eqref{opt1: getting optimal combined pairs} and \eqref{opt2: getting optimal gaussian pairs} are determined by a derivative-free Mesh Adaptive Direct Search (MADS) solver \citep{audet_mesh_2006}. Besides, the proposed overbound contains the CGCM-component parameters (i.e., ($m_o$, $\lambda_o$)) from the optimal paired CGCM overbound and the Gaussian-component parameters (i.e., ($\mu_o$, $\sigma_o$)) from the optimal paired Gaussian overbound.

\begin{table}[htb]
 \caption{{Overbounding parameters of NavDEN, Cauchy-Gaussian and two-step Gaussian methods for the urban dataset (n.s.u. profile). The tilde notations for NavDEN overbound indicate that the parameters are divided by the fundamental spacing $\Delta$ to be unitless.}}
 \label{tab: real nsu2 paras}
 %text alignment: l -left; c - center; r -right
{
\begin{tblr}{colspec={X[c]X[c]X[c]},
width=\textwidth,
row{even} = {white,font=\small},
row{odd} = {bg=black!10,font=\small},
row{1} = {bg=black!20,font=\bfseries\small},
hline{Z} = {1pt,solid,black!60},
rowsep=3pt
}
Overbounding technique & Parameters & Values  \\
NavDEN overbound & \makecell[c]{$\Delta$\\ $\tilde{x}_{max}$\\ $\tilde{x}_{min}$\\ $\tilde{B}$\\ $\tilde{C}$ \\ $k_{tr}$ \\ $k_{max}$ \\ $k_{min}$ \\ $k_{bias}$ } & \makecell[c]{0.2m \\ 42 \\ -42 \\ 50 \\ 130 \\ 8 \\ 32 \\ -33 \\ 4} \\ 
Two-step Gaussian overbound & $(\mu_{2s}, \sigma_{2s})$ & (4.50m, 3.54m) \\
Cauchy-Gaussian overbound & \makecell[c]{CGCM component $(m_o, \lambda_o)$ \\ Gaussian component $(\mu_o, \sigma_o)$} & \makecell[c]{(0.05m, 0.51m) \\ (5.00m, 1.71m)}
\end{tblr}}
\end{table}

{Figure \ref{fig: real nsu2 ovb}b displays the left and right halves of NavDEN overbound and presents the two-step Gaussian overbound and Cauchy-Gaussian overbound in the form of single-CDF analog defined by Equation \eqref{equ: single-cdf analog}. As can be seen, each half of NavDEN overbound is constructed using a discrete model, while that of the other two overbounding methods is formulated to be continuous. The two-step Gaussian overbound is visibly the most conservative of the three methods. The Cauchy-Gaussian overbound has the tightest bounding performance in the central region (absolute error within 2m), and the inset plot confirms it successfully bounds the empirical bias. Notably, both the two-step overbound and the Cauchy-Gaussian overbound include a Gaussian component.} However, with the assistance of the Cauchy distribution in the core region, the CGCM component in the proposed overbound generates a significantly smaller location parameter ($m_o=$ 0.05m) than the two-step Gaussian overbound ($\mu_{2s}=$ 4.50m). This small location parameter $m_o$ enables the resultant Cauchy-Gaussian overbound to yield considerably tighter bounding in the core region, as illustrated in the analog single-CDF overbounds of Figure \ref{fig: real nsu2 ovb}b.

{To analyze the bounding performance at tail regions, we plot the Log CDFs in Figure \ref{fig: real nsu2 ovb}c and Log CCDFs in Figure \ref{fig: real nsu2 ovb}d to demonstrate the left and right tails of the overbounding distributions, respectively. Both figures demonstrate the impressively tight bounding of the NavDEN model for absolute errors between 2m and 11m. However, it becomes less sharp than the proposed method in the far tails (absolute error beyond 12m) due to its heavy-tailed decay. Additionally, while both the two-step Gaussian overbound and the Cauchy-Gaussian overbound use a Gaussian profile at the far-tail region, the proposed method produces a sharper bound, as depicted in both Figure \ref{fig: real nsu2 ovb}c and Figure \ref{fig: real nsu2 ovb}d. This is because the Gaussian component of the Cauchy-Gaussian overbound has a smaller scale parameter ($\sigma_o=$ 1.71m) than the two-step Gaussian overbound ($\sigma_{2s}=$ 3.54m).}

{An additional analysis is conducted to compare the position-domain bounding performance of the different overbounding methods. Figure \ref{fig: real nsu2 pl} shows the time series of VPEs, where the mean value is 0.41m and the maximum reaches 2.89m. In each epoch, we calculate the VPL of the NavDEN, two-step Gaussian, and Cauchy-Gaussian overbounds via Equation \eqref{equ: VPL}, and plot the results in the same figure. Notably, the VPL is calculated based on the worst-case DGNSS error distribution, which is given by the error for elevation angles observed from 30$-$35$^{\circ}$. Figure \ref{fig: real nsu2 pl} shows that all three methods produce VPL that successfully bound the VPE in each epoch. The Cauchy-Gaussian overbound consistently yields the lowest VPLs due to its advantageous fit to the empirical error distribution at both core and tail regions. The two-step Gaussian overbound produces the highest VPLs, while the NavDEN VPLs are at a moderate level. Table  \ref{tab: real nsu2 pl improve} quantifies the percentage of reduction in VPL achieved by the Cauchy-Gaussian overbound compared to the other two overbounds. On average, the VPL is reduced by 21.07\% compared to the NavDEN overbound and by 47.66\% compared to the two-step Gaussian overbound. The VPL reduction is consistent across all metrics, with a maximum percentage reduction of nearly 50\% against the two-step Gaussian overbound. These results indicate that the proposed method can effectively reduce the VPL for n.s.u errors without compromising integrity.}

% The VPL is reduced by approximately 15\% on average, and by about 20\% as a maximum. { Besides, the Q1 (12.57\%), Q2 (15.79\%), and Q3 (17.08\%) of the VPL reductions are provided.} These results indicate that the proposed method can effectively reduce the VPLs for n.s.u errors without compromising integrity.

\begin{figure}[t]
    \centering
    \includegraphics[width = 0.95\textwidth]{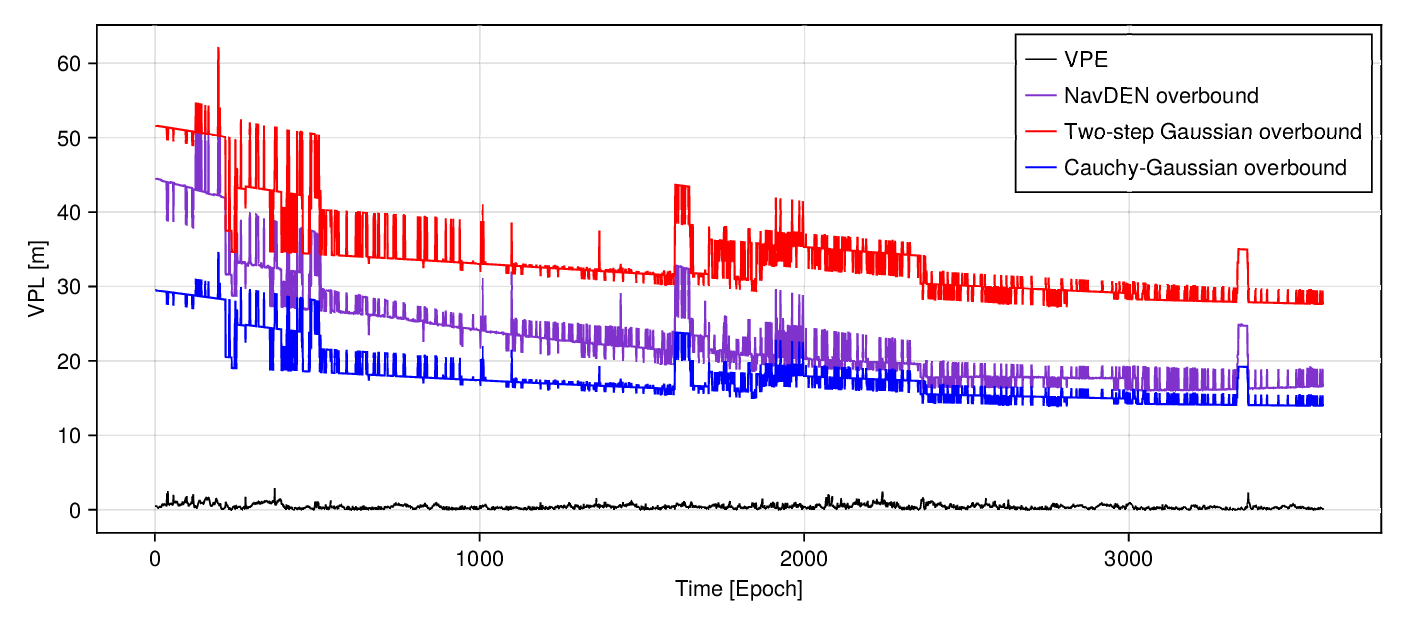}
    \caption{{VPL of the NavDEN, two-step Gaussian and Cauchy-Gaussian overbounding methods, with $P_{HMI} = 1\times 10^{-9}$. The VPEs are also plotted for reference.}}
    \label{fig: real nsu2 pl} 
\end{figure}

% \begin{figure}[t]
%     \centering
%     \includegraphics[width = 0.995\textwidth]{Figures/exp/vpl_seaside.eps}
%     \caption{{VPL of the NavDEN, two-step Gaussian and Cauchy-Gaussian overbounding methods, with $P_{HMI} = 1\times 10^{-9}$. The VPEs are also plotted for reference.}}
%     \label{fig: real nsu1 pl} 
% \end{figure}

\begin{table}[htb]
 \caption{{Reduction in VPL by the Cauchy-Gaussian overbound compared to the two-step Gaussian overbound in the urban dataset 2. Q1, Q2, and Q3 represent the 25th, 50th and 75th percentiles, respectively.}}
 \label{tab: real nsu2 pl improve}
 %text alignment: l -left; c - center; r -right
{
\begin{tblr}{colspec={X[c]X[c]X[c]X[c]X[c]X[c]X[c]},
width=\textwidth,
row{even} = {white,font=\small},
row{odd} = {bg=black!10,font=\small},
row{1} = {bg=black!20,font=\bfseries\small},
hline{Z} = {1pt,solid,black!60},
rowsep=3pt
}
Percentage of reduction & Average & Maximum  & Minimum & Q1 & Q2 & Q3\\
Compared to NavDEN overbound & 21.07\% & 40.14\% & 21.07\% & 14.31\% & 21.11\% & 27.36\% \\
Compared to two-step Gaussian overbound & 47.66\% & 49.67\% & 41.50\% & 46.99\% & 48.16\% & 49.14\%
\end{tblr}}
\end{table}

\subsection{Impact of heavy-tailedness}
\label{sec: analyze heavy-tailedness}
This section explores how the heavy-tailedness in the n.s.u. error distribution impacts the positioning bounding performance of the proposed Cauchy-Gaussian overbound. {For comparison, we also evaluated the performance of the two-step Gaussian overbound, a method supported by an open-source toolkit for heuristically bounding errors.} A biased BGMM with the following setting is used to simulate the n.s.u. error distribution:
\begin{equation}
    f_e(x) = p_1 f_G(x;0.1, 1) + (1-p_1) f_G(x;0.1, 10),
    \label{eq:biasBGMM}
\end{equation}
where the location parameter 0.1m in both Gaussian components directly contributes to the bias of 0.1m in the resultant error distribution, and $p_1$ represents the proportion of the first component. Notably, as the $p_1$ increases, the resultant distribution approximates closer to the first component $f_G(x;0.1, 1)$, which is more centralized than the second component $f_G(x;0.1, 10)$. However, the second component with a larger scale parameter (10m) always enables a fraction of data samples to deviate largely away from the center location (i.e., 0.1m). Therefore, the increase in $p_1$ can enlarge the kurtosis (or say, heavy-tailedness) of the BGMM error distribution.

Similar to the setting in Section \ref{sec: exp su errors}, we use MAAST \citep{jan_matlab_2001} to simulate satellite positions every 100 seconds over 24 hours. The receiver and reference locations are listed in Appendix \ref{app 0: LLH, MAAST}. The DGNSS measurements are generated by adding the randomly generated sample from the error distribution in Equation \eqref{eq:biasBGMM} to the true differential range. In the experiment, $p_1$ is set from 0.60 to 0.95 with an increment of 0.05. For each value of $p_1$, we calculate the average VPL of both overbounding methods throughout the time frame.

Figure \ref{fig: analyze bias 1} depicts the average VPL given by the two-step Gaussian and Cauchy-Gaussian overbounding methods, separately. As expected, the proposed overbound generates less conservative bounding than the two-step Gaussian overbound at both core and tail regions, thus yielding lower average VPLs across different settings in $p_1$. Besides, both curves of average VPL calculations show the overall descending trend, with fluctuations induced by several extreme random data samples. The two-step Gaussian overbound generates the average VPL slightly decreasing from 33.80m to 31.84m, when $p_1$ increases from 0.60 to 0.95. A possible reason is that the resultant distribution $f_e$ gradually allocates less proportion to the second component, which has a fatter tail than the first component. Correspondingly, the scale parameters of the Gaussian-based overbound are slightly reduced, which enables tighter bounding and thus lower VPL calculations. In contrast, the curve of Cauchy-Gaussian overbound decreases more significantly from 31.60m to 24.22m as $p_1$ increases. This can be explained by the fact that the empirical distribution becomes more heavy-tailed with the increase of $p_1$ and possesses longer tails and a sharper core region. Apart from the benefit from a lower scale parameter given by the optimal paired Gaussian, the Cauchy-Gaussian overbound can advantageously bound the core region using a smaller location parameter by leveraging the optimal paired CGCM, thus resulting in smaller average VPLs. Figure \ref{fig: analyze bias 2} shows the percentage of reduction in average VPLs by the Cauchy-Gaussian overbound compared to the two-step Gaussian overbound against different settings of $p_1$. As can be seen, the curve rises gradually from 6.53\% to 23.94\% when $p_1$ reaches 0.95. The result also potentially reveals the strength of the proposed method to less conservatively bound increasingly heavy-tailed distributions, compared to the two-step Gaussian overbound.

\begin{figure}[t]
        \centering
        \begin{subfigure}{0.49\textwidth} % width for the subfigure
            \includegraphics[width=0.95\textwidth]{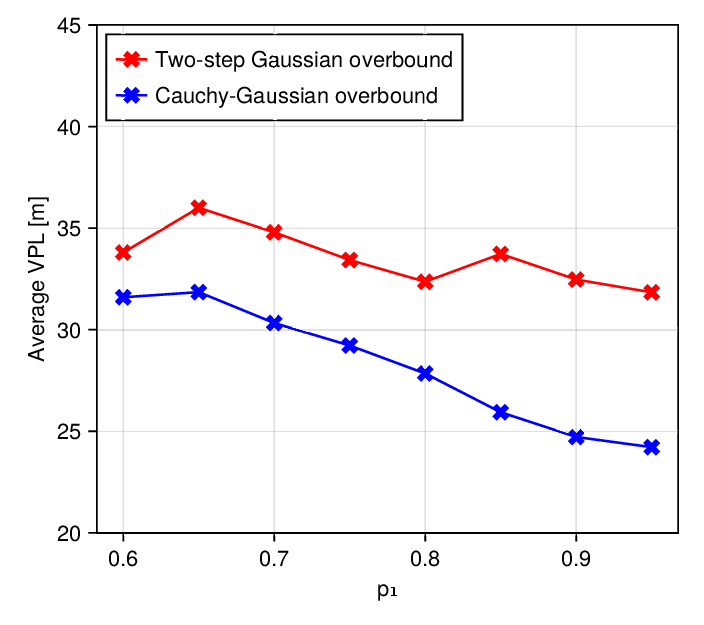}
            \caption{\centering}
            \label{fig: analyze bias 1}
        \end{subfigure}
        \hfill 
        \begin{subfigure}{0.49\textwidth} 
            \includegraphics[width=0.95\textwidth]{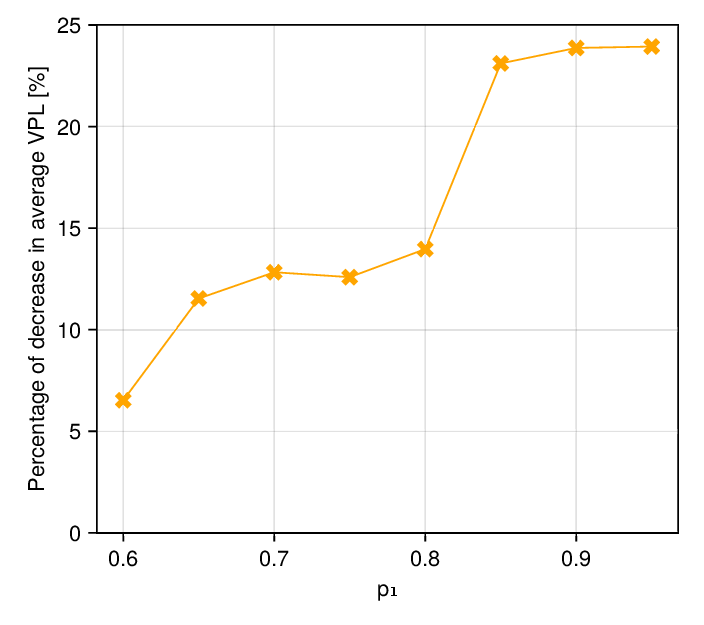}
            \caption{\centering}
            \label{fig: analyze bias 2}
        \end{subfigure}
        \caption{Results for different biased and heavy-tailed BGMM: (a) average VPLs given by the two overbounding methods; (b) the percentage of reduction in average VPL yielded by Cauchy-Gaussian overbound, compared to the two-step Gaussian overbound.}
        \label{fig: analyze bias}
\end{figure}

\section{Conclusion and Future Work}
\label{sec: conclusion and outlook}
{This work employs the simply parameterized Cauchy distribution to characterize the heavy-tailed properties of empirical errors. Building on this, the proposed Cauchy-Gaussian overbound is designed to tightly bound both symmetric unimodal (s.u.) and not symmetric unimodal (n.s.u.) heavy-tailed error distributions.} For both error types, a systematic three-step procedure is developed to determine the optimal parameters of the overbound. The proposed overbound offers a single-CDF bound for s.u. distributions and paired bounds for n.s.u. error distributions. For s.u. errors, the Cauchy-Gaussian overbound employs zero-mean Cauchy overbound at the core region and zero-mean Gaussian overbound at the tail regions. In this process, the determined overbounding relationship between centrally aligned Cauchy and Gaussian distributions is applied to reduce the computational load when searching for the optimal scale parameter of the Cauchy overbound. Subsequently, a tangential transition from Cauchy to Gaussian is designed to maintain the overbounding preservation through convolution. For n.s.u. errors, the empirical distributions are bounded with an optimal pair of half-Cauchy-half-Gaussian CGCM, which can generate significantly smaller location parameters than pure Cauchy or pure Gaussian distribution. An optimal paired Gaussian overbound is determined subsequently. The paired CGCM overbound and paired Gaussian overbound are then synthesized to produce piecewise paired overbounding distributions to bound the original n.s.u errors tightly. The bounding performance of the Cauchy-Gaussian overbound is compared with the single-CDF Gaussian overbound for s.u. errors through a simulated DGNSS dataset. {For n.s.u. errors, the performance is evaluated in comparison with the NavDEN overbound and two-step Gaussian overbound through a real DGNSS dataset collected in the Hong Kong urban environment.} For both error types, the proposed Cauchy-Gaussian overbound provides overall tighter bounds than the baseline methods. {In the position domain, the Cauchy-Gaussian method reduces the VPL by approximately 15\% on average for s.u. errors. For n.s.u. errors, the average VPL reduction is even more significant, reaching 21\% compared to the NavDEN overbound and over 47\% compared to the two-step Gaussian overbound.} Besides, as heavy-tailedness increases in the error distributions, the Cauchy-Gaussian overbound yields a more significant reduction in average VPL.

Nevertheless, the Cauchy-Gaussian overbound has more than three parameters for both s.u. and n.s.u. empirical distributions, which pose challenges for broadcasting overbound parameters in augmentation systems with limited bandwidth channels. Possible solutions involve developing a quantitative relationship between the Gaussian component and the Cauchy-based component (or the CGCM component), such that only one pair of location and scale parameters is needed in overbound computation. Besides, the optimal location and scale parameters of the paired CGCM overbound and paired Gaussian overbound can be impacted by the objective function. It is thus worthwhile to analyze the bounding performance yielded by the current design compared to other cost forms, such as the discretized Jensen-Shannon divergence. {To more conveniently compute tight bounds against heavy-tailed n.s.u. errors, extended work could utilize a two-step architecture \citep{blanch_gaussian_2018} where the Cauchy-Gaussian model is used to overbound an intermediate s.u. envelope generated from empirical n.s.u. errors.} Future directions may also include applying the proposed Cauchy-Gaussian overbound to fault detection algorithms in scenarios of both s.u. and n.s.u. measurement errors.

\newpage
\appendix
\section{Coordinates of receivers and reference stations}
\label{app 0: LLH, MAAST}
\begin{table}[htb]
 \caption{Locations (longitude, latitude, height) of receivers and reference stations in the simulation and urban dataset. It is noted that the simulation in Sections \ref{sec: exp su errors} and \ref{sec: analyze heavy-tailedness} share the same pair of receiver and reference station locations}
 \label{atab: locations}
 %text alignment: l -left; c - center; r -right
\begin{tblr}{colspec={X[c]X[c]X[c]},
width=\textwidth,
row{even} = {white,font=\small},
row{odd} = {bg=black!10,font=\small},
row{1} = {bg=black!20,font=\bfseries\small},
hline{Z} = {1pt,solid,black!60},
rowsep=3pt
}
Settings & Simulation & Urban dataset validation  \\
Receiver location & (60.00 deg, 154.10 deg, 1.08 $\times $10$^{-7}$ m) & (22.30 deg, 114.18 deg, 2.08 m)   \\
Reference station location & (60.00 deg, 154.00 deg, 1.06 $\times $10$^{-7}$m) & (22.32 deg, 114.14 deg, 20.24 m)  \\
Distance & 5.58 km & 4.73 km
\end{tblr}
\end{table}

\section{Sufficient and necessary conditions of Cauchy overbound on a centrally-aligned Gaussian distribution}
\label{app 1: cauchy ovb on a known gaussian}
The PDF and CDF of the standard Gaussian distribution with location parameter $\mu$ and scale parameter $\sigma$ are shown as:
\begin{align}
    f_G(x; \mu, \sigma) &= \frac{1}{\sigma\sqrt{2\pi}}\exp\left({-\frac{1}{2}\left(\frac{x-\mu}{\sigma}\right)^2}\right), \label{aequ: gaussian pdf}\\
    F_{G}(x; \mu, \sigma) &= \frac{1}{2} \left(1+ \text{erf}\left(\frac{x-\mu}{\sigma\sqrt{2}}\right) \right), \label{aequ: gaussian cdf}
\end{align}
where $\text{erf}(\cdot)$ stands for the error function. Similarly, the PDF and CDF of the Cauchy distribution with location parameter $m$ and scale parameter $\lambda$ are given by:
\begin{align}
    f_C(x; m, \lambda)&=\frac{1}{\lambda \pi}\frac{1}{1+(\frac{x-m}{\lambda})^2},  \label{aequ: cauchy pdf}\\
    F_C(x; m, \lambda)&=\frac{1}{2}+\frac{1}{\pi}\arctan \left(\frac{x-m}{\lambda}\right). \label{aequ: cauchy cdf}
\end{align}

\begin{theorem} \label{thm: cauchy ovb gaussian}
    For a Gaussian model $N(\mu,\sigma)$ and a Cauchy model $C(m, \lambda)$, if their medians are aligned (i.e., $\mu=m=M_0$, where $M_0$ denotes the known value of the median), then the sufficient and necessary condition of $C(M_0,\lambda)$ overbounding $N(M_0,\sigma)$ is $\lambda \geq \sqrt{\frac{2}{\pi}}\sigma$. 
\end{theorem}

It should be highlighted that the overbounding properties of symmetric unimodal (s.u.) distributions in Section \ref{sec: sym cdf ovb} require both the error and overbounding distributions to be zero-mean (i.e., $\mu=m=0$) \citep{decleene_defining_2000}. Here, an extended definition of overbounding relationship between centrally-aligned Cauchy ($F_C$) and Gaussian ($F_G$) distributions is provided:
\begin{align}
    F_G(x; M_0, \sigma) \leq  F_C(x; M_0, \lambda)\quad\forall x \leq M_0, \label{aequ: sym cdf x<0}\\
    F_G(x; M_0, \sigma) \geq   F_C(x; M_0, \lambda)\quad\forall x > M_0, \label{aequ: sym cdf x>0}
\end{align}
where $M_0 \in \mathbb{R}$. \par

For the convenience of proving the Theorem, a subtraction function $t(x)$ between the CDFs is introduced as:
\begin{equation}
    t(x) = F_G(x; M_0, \sigma)-F_C(x; M_0, \lambda),
\end{equation}
then its derivative function further gives the subtraction relationship between the PDFs:
\begin{equation}
    t'(x) = f_G(x; M_0, \sigma)-f_C(x; M_0, \lambda).
\end{equation}
If the Cauchy distribution overbounds the Gaussian distribution, their CDFs will coincide at $(M_0,1/2)$, which means
\begin{equation}
     t(M_0) =0, \label{aequ: t(M0)=0}
\end{equation}
and the definition of Equation \eqref{aequ: sym cdf x<0} and \eqref{aequ: sym cdf x>0} further gives:
\begin{align}
    t(x) \leq 0 \quad \forall x \leq M_0, \label{aequ: sym cdf x<}\\
    t(x) > 0 \quad \forall x > M_0. \label{aequ: sym cdf x>} 
\end{align}
Since the Cauchy PDF has super-exponential tails, there is: 
\begin{align}
    t'(-\infty)&< 0, \label{aequ: t'(-inf)<0}\\
     t'(+\infty)&< 0. \label{aequ: t'(+inf)<0}
\end{align}
Due to the symmetry of both Cauchy and Gaussian distributions, the following paragraphs will mainly analyze the case when $x\leq M_0$ because similar procedures can be expanded to $x> M_0$. As such, the proofs of both sufficiency and necessity for the Theorem will be provided as follows.

\begin{proof}[\textbf{Proof of sufficiency}] \label{sufficiency}  (if $\lambda \geq \sqrt{\frac{2}{\pi}}\sigma$, then $C(M_0,\lambda)$ overbounds $N(M_0,\sigma)$.)

In the domain of $x\leq M_0$, according to the properties listed from Equation \eqref{aequ: t(M0)=0} to \eqref{aequ: sym cdf x>} and the heavy-tailedness of Cauchy distribution mentioned in Equation \eqref{aequ: t'(-inf)<0} and \eqref{aequ: t'(+inf)<0}, the conclusion \enquote{$C(M_0,\lambda)$ overbounds $N(M_0,\sigma)$} is equivalent to \enquote{$\nexists x^{*}\in(-\infty, M_0)$ such that $t(x^{*})>0$}, and we only need to disprove the opposite statement, which specifically gives \enquote{$\exists x^{*}\in(-\infty, M_0)$ such that $t(x^{*})>0$}.

If $\lambda \geq \sqrt{\frac{2}{\pi}}\sigma$, it can be proved that there is one and only one intersection between $f_C(x)$ and $f_G(x)$ when $x< M_0$ (details can be seen in the Lemma). There exists $x_0\in (-\infty, M_0)$ such that $t'(x_0)=0$, $t'(x_0^+)>0 $, and $ t'(x_0^-)<0$. In this case, when $x\in(-\infty,x_0), t'(x)<0$, and $ t(x)$ is decreasing; when  $x\in(x_0,M_0), t'(x)>0$, and $ t(x)$ is increasing.  Besides, it is inferred that $x^*>x_0$ (otherwise the decreasing $t(x)$ in $ (-\infty,x_0)$ gives $t(x^*)<t(-\infty)<0$, which is contradictory to the assumption $t(x^*)=0$). Since $t(x)$ is increasing in $(x_0, M_0)$, and that $x^*\in(x_0,M_0)$, there will be $t(M_0)>t(x^*)=0$, which is contradictory to the overbounding properties $t(M_0)=0$. 

Therefore, the opposite conclusion \enquote{$\exists x^{*}\in(-\infty, M_0)$ such that $t(x^{*})=0$} is untenable, which proves that \enquote{$\nexists x^{*}\in(-\infty, M_0)$ such that $t(x^{*})>0$} and $C(M_0,\lambda)$ overbounds $N(M_0,\sigma)$. This contradiction can likewise be extended to the case when $x > M_0$. The sufficiency is proved.
\end{proof}

\begin{lemma} \label{lemma1}
    If $\lambda \geq \sqrt{\frac{2}{\pi}}\sigma$, then $f_C(x;M_0,\lambda)$ and $f_G(x;M_0;\sigma)$ have one and only one intersection when $x< M_0$.
\end{lemma}

\begin{proof}
Since $f_G(x)>0$ $\forall x\in \mathbb{R}$ and $f_C(x)>0$ $\forall x\in \mathbb{R}$ and they are both continuous function, we introduce the function
\begin{equation}
  g(x)=\frac{f_G(x)}{f_C(x)} -1 =\frac{\lambda}{\sigma} \sqrt{\frac{\pi}{2}} \exp\left({-\frac{1}{2}\left(\frac{x-M_0}{\sigma}\right)^2}\right) \cdot \left( 1+\left(\frac{x-M_0}{\lambda}\right)^2\right) -1.
\end{equation}
Based on the PDF expressions of Gaussian and Cauchy in Equation \eqref{aequ: gaussian pdf} and \eqref{aequ: cauchy pdf}, we know that
\begin{align}
    f_G(x)&= O(\exp\left({-x^2}\right))\quad as\ |x| \rightarrow \infty, \\
    f_C(x)&= O(x^{-2}) \quad as\ |x| \rightarrow \infty,
\end{align}
where $O(\cdot)$ is the big O notation, then it is inferrable that
  \begin{align}
      g(-\infty) = -1, \\
      g(+\infty) = -1.
  \end{align}
    Denote $u=x-M_0$, then $u\in(-\infty, 0)$, then $f_C(x;M_0,\lambda)$ and $f_G(x;M_0;\sigma)$ have an intersection means that $g(u)=0$ has a solution. The Lemma further changes to \enquote{If $\lambda \geq \sqrt{\frac{2}{\pi}}\sigma$, then $g(u)$ has one and only one zero point in $(-\infty, 0)$}. Besides, the simplified $g(u)$ gives:
\begin{equation}
        g(u) =\frac{\lambda}{\sigma} \sqrt{\frac{\pi}{2}} \exp\left({-\frac{1}{2}\left(\frac{u}{\sigma}\right)^2}\right) \cdot \left( 1+\left(\frac{u}{\lambda}\right)^2\right) -1,
    \end{equation}
and its derivative with respect to $u$ gives:
\begin{align}
    g'(u) & = \frac{\lambda}{\sigma} \sqrt{\frac{\pi}{2}}  \cdot \left( -\frac{u}{\sigma^2} \exp\left({-\frac{1}{2}\left(\frac{u}{\sigma}\right)^2}\right) \cdot \left( 1+\left(\frac{u}{\lambda}\right)^2\right)+\frac{2u}{\lambda^2}  \exp\left({-\frac{1}{2}\left(\frac{u}{\sigma}\right)^2}\right) \right) \nonumber \\
    & = \frac{\lambda}{\sigma} \sqrt{\frac{\pi}{2}} \exp\left({-\frac{1}{2}\left(\frac{u}{\sigma}\right)^2}\right) \cdot \left( -\frac{u}{\sigma^2}-\frac{u^3}{\sigma^2\lambda^2} +\frac{2u}{\lambda^2}\right)  \nonumber \\
    & = \sqrt{\frac{\pi}{2}} \frac{\exp\left({-\frac{1}{2}(\frac{u}{\sigma})^2}\right)}{\sigma^3\lambda}  \cdot \left( -u^3-(\lambda^2 - 2\sigma^2)u \right).
\end{align}
    
    If $\lambda^2 - 2\sigma^2 \geq 0$ which means $\lambda \geq \sqrt{2} \sigma$, $g'(u) > 0$ thus $g(u)$ is increasing in $(-\infty, 0)$. Since $g(u)$ is continuous, $g(u)_{max}=g(0^-)\to g(0)=\frac{\lambda}{\sigma} \sqrt{\frac{\pi}{2}} -1 \geq \sqrt{\pi}-1 > 0$. Combined with $g(-\infty)=-1<0$, $g(u)$ has one and only one zero point on $(-\infty, 0)$. 
    
    If $\lambda^2 - 2\sigma^2 < 0$ which means $\sqrt{\frac{2}{\pi}} \sigma \leq \lambda < \sqrt{2} \sigma$, when $u\in(-\sqrt{2\sigma^2-\lambda^2}, 0)$, $g'(u)<0$ thus $g(u)$ is decreasing in this domain, $g(u)_{min}=g(0^-) \to g(0) = \frac{\lambda}{\sigma} \sqrt{\frac{\pi}{2}} -1 \geq  0$, $g(u)$ doesn't have any zero point on $(-\sqrt{2\sigma^2-\lambda^2}, 0)$. When $u\in(-\infty, -\sqrt{2\sigma^2-\lambda^2})$, $g'(u)>0$ thus $g(u)$ is increasing in this domain, $g(u)_{max}=g(-\sqrt{2\sigma^2-\lambda^2}) > g(0) \geq 0$. Given that $g(x)$ is continuous and $g(-\infty)=-1<0$, $g(u)$ has only one zero point on $(-\infty, -\sqrt{2\sigma^2-\lambda^2})$.
    
    Combining the cases of both $\lambda^2 - 2\sigma^2 \geq 0$ and $\lambda^2 - 2\sigma^2 < 0$, we can conclude that $g(u)$ has one and only one zero point in $(-\infty, 0)$. This proves that $f_C(x;M_0,\lambda)$ and $f_G(x;M_0;\sigma)$ have one and only one intersection in $(-\infty, M_0)$.
\end{proof}

\begin{proof}[\textbf{Proof of necessity}] \label{necessity}  (if $C(M_0,\lambda)$ overbounds $N(M_0,\sigma)$, then $\lambda \geq \sqrt{\frac{2}{\pi}}\sigma$.)\par
    Similarly, we assume the opposite conclusion $\lambda < \sqrt{\frac{2}{\pi}}\sigma$ is true. It is known that the maximum of PDF for both Cauchy and Gaussian distributions happens when $x=M_0$. Therefore, $f_C(x)_{max}=\frac{1}{\lambda \pi}>\frac{1}{\sigma\sqrt{2\pi}}=f_G(x)_{max}$. 
    
    When $x\leq M_0$, there will always be a positive $\epsilon$ such that when $ x\in (M_0-\epsilon, M_0)$, $ t'(x)=f_G(x; M_0, \sigma)-f_C(x; M_0, \lambda)<0$ thus $t(x)$ is decreasing in this domain. Based on the overbounding property in Equation \eqref{aequ: sym cdf x<}, it is inferrable that $ t(M_0-\epsilon) < 0.$ Now that $M_0-\epsilon$ and that $t(x)$ is decreasing, there is $t(M_0)<t(M_0-\epsilon)<0$, which is contradictory to the premise  $t(M_0)=0$ mentioned in Equation \eqref{aequ: t(M0)=0}. 
    
    The contradiction can likewise be extended to the case when $x > M_0$, thus the assumption $\lambda < \sqrt{\frac{2}{\pi}}\sigma$ is untenable and it can be concluded that $\lambda < \sqrt{\frac{2}{\pi}}\sigma$. The necessity is proved.
\end{proof}

{
\section{Proof of existence of a tangential line between Cauchy and Gaussian overbounds}
\label{app 1.5: tangent existence}
In this study, the tangential link from Cauchy overbound to Gaussian overbound is always possible as long as the Cauchy CDF is not an overbound for the Gaussian CDF. The statement contains two parts, including: 1) sufficiency: ``if the Cauchy CDF is not an overbound for the Gaussian CDF, then a tangential link exists"; 2) necessity: ``if a tangential link exists; then the Cauchy CDF is not an overbound for the Gaussian CDF". For the sake of notation, we denote the CDFs of the Gaussian and Cauchy overbounds to be $F_G(x)$ and $F_C(x)$, respectively.  Besides, we will focus on the positive half (i.e., when $x>0$), for instance, and similar procedures of proof can be inferred for the negative half (i.e., when $x<0$) due to the symmetry of CDFs. In the sequel, sufficiency and necessity will be proven, respectively, in the case of bounding symmetric and unimodal errors. \par

\begin{figure}[ht]
    \centering
    % First subfigure
    \begin{subfigure}[b]{0.49\textwidth}
        \centering
        \includegraphics[width=0.95\textwidth]{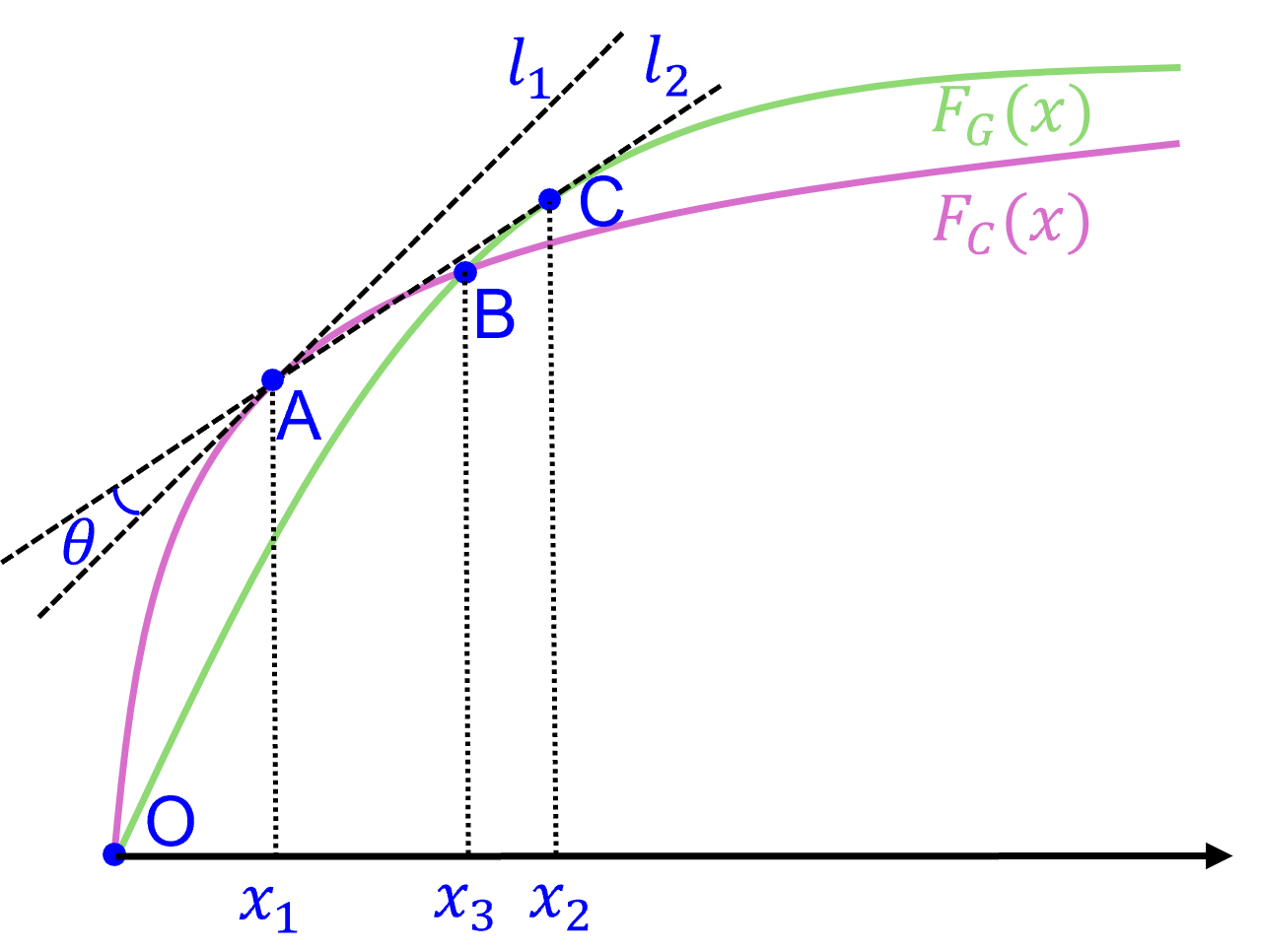}
        \caption{}
        % \label{fig:subfig1}
    \end{subfigure}
    \hfill
    % Second subfigure
    \begin{subfigure}[b]{0.49\textwidth}
        \centering
        \includegraphics[width=0.95\textwidth]{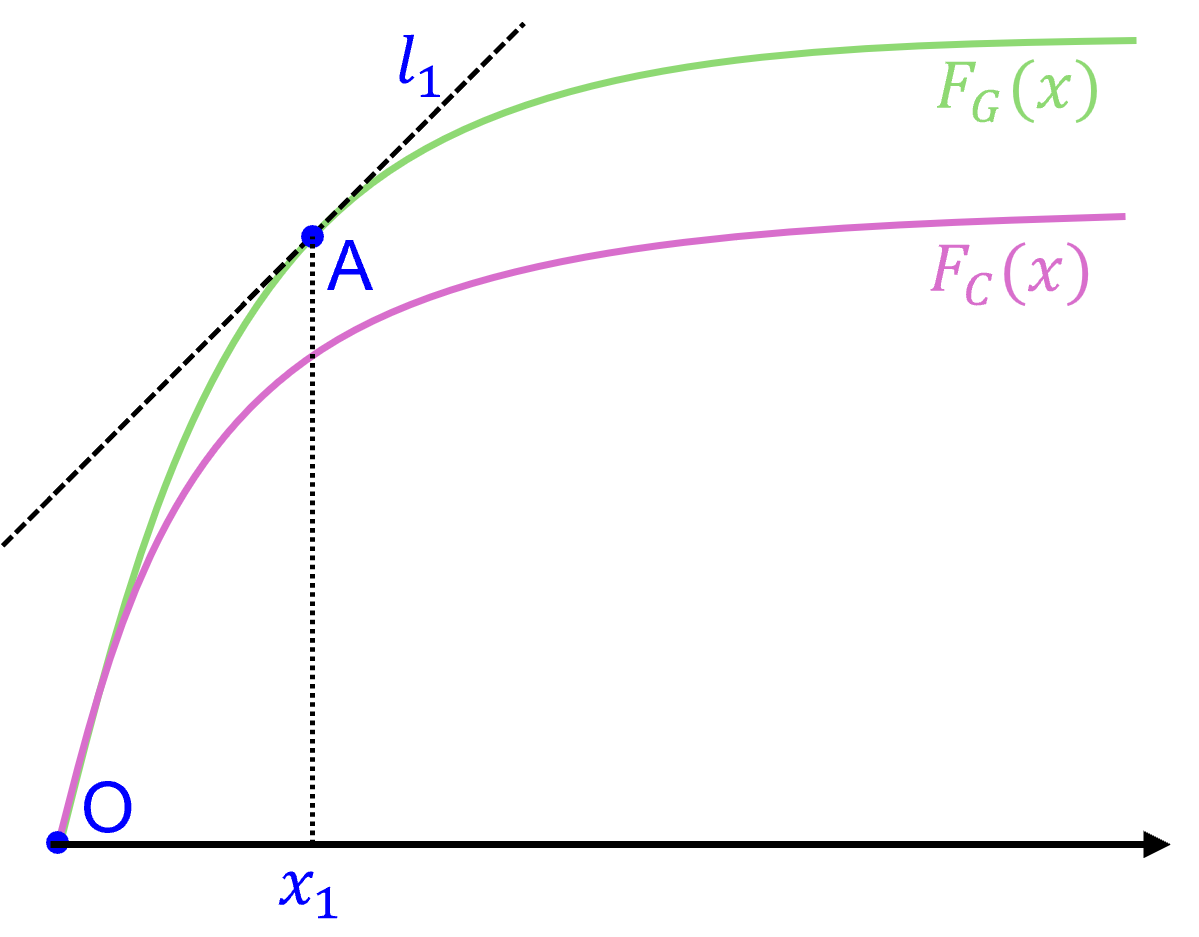}
        \caption{}
        % \label{fig:subfig2}
    \end{subfigure}
    \caption{Schematics, including (a) for proof of sufficiency and (b) for proof of necessity}
    \label{fig: tangential link}
\end{figure}

\begin{proof}[Proof of sufficiency]
We may first assume the opposite of the conclusion (i.e., a tangential link between Cauchy and Gaussian CDFs does not exist). When $x>0$, both the Cauchy CDF and Gaussian CDF are strictly concave because their second-order derivatives are always negative. Since the hypothesis states that Cauchy CDF is not an overbound for Gaussian CDF, it is inferable that an intersection (denoted as point B in Fig. \ref{fig: tangential link}a) always exists between the two CDFs when $x>0$, such that
\begin{align}
    & F_G(x) < F_C(x),\quad \forall 0<x<x_3, \\
    & F_G(x) > F_C(x),\quad \forall x>x_3. \label{equ: r2_sufficiency_1}
\end{align}
Due to the concavity of CDF, for any $x_1\in(0,x_3)$, passing through point A$(x_1,F_C(x_1))$, a unique tangential line $l_1$ for $F_C(x)$ can always be drawn:
\begin{equation}
    l_1:\quad y=F_C'(x_1)(x-x_1)+F_C(x_1).
\end{equation}
Besides, $F_C(x_1)>F_G(x_1)$, thus point A is above the Gaussian CDF. Passing through point A, a tangential line $l_2$ for the concave function $F_G(x)$ can always be drawn, with a tangent point C $(x_2,F_G(x_2))$: 
\begin{equation}
    l_2:\quad y=F_G'(x_2)(x-x_2)+F_G(x_2).
\end{equation}
Given by the premise that a tangential link between the two CDFs does not exist, we know that $l_1$ and $l_2$ cannot coincide. Equivalently, this gives that the angle $\theta$ formed by $l_1$ (tangent to $F_C(x)$) and $l_2$ (tangent to $F_G(x)$) cannot be 0. We may obtain that
\begin{equation}
    \tan(\theta) = \left|\frac{F_C'(x_1) - F_G'(x_2)}{1+F_C'(x_1) F_G'(x_2)}\right|\neq 0,
\end{equation}
which further gives that 
\begin{equation}
    F_C'(x_1) \neq F_G'(x_2),\quad \forall x_1\in(0,x_3),\forall x_2\in(x_3, \infty). \label{equ: r2_sufficiency_1a}
\end{equation}
Since the strict concavity of the two CDFs, both $F_C'(x)$ and $F_G'(x)$ are decreasing functions when $x>0$. Therefore, 
\begin{align}
     &F_C'(x_1)\in \left(F_C'(x_1)_{\min} ,\ F_C'(x_1)_{\max} \right) = \left(F_C'(x_3),\ F_C'(0) \right), \label{equ: r2_sufficiency_2} \\
      &F_G'(x_2)\in \left(F_G'(x_2)_{\min} ,\ F_G'(x_2)_{\max} \right) = \left(\lim_{x\rightarrow +\infty}F_G'(x),\ F_G'(x_3) \right) =\left(0,\ F_G'(x_3) \right),
\end{align}
where
\begin{equation}
   F_C'(0) > F_G'(0) >F_G'(x_3).
\end{equation}
Besides, it can be further verified that
\begin{equation}
     F_C'(x_3)<F_G'(x_3). \label{equ: r2_sufficiency_3}
\end{equation}
Let's first assume the opposite $F_C'(x_3)\geq F_G'(x_3)$. The intersection condition at $x=x_3$ gives $F_C(x_3)=F_G(x_3)$. There exists a small positive $\epsilon$ such that
\begin{equation}
    F_C(x_3+\epsilon)=F_C(x_3)+\epsilon *F_C'(x_3)=F_G(x_3)+\epsilon *F_C'(x_3) \geq F_G(x_3)+\epsilon *F_G'(x_3)=F_G(x_3+\epsilon),
\end{equation}
which contradicts with the property in Equation \eqref{equ: r2_sufficiency_1}. \par
Combining Equation \eqref{equ: r2_sufficiency_2} to \eqref{equ: r2_sufficiency_3}, we obtain
\begin{equation}
    \min\left(F_G'(x_2)\right)<\min\left(F_C'(x_1)\right) <  \max\left(F_G'(x_2)\right)<\max\left(F_C'(x_1)\right). 
\end{equation}
Therefore, we can conclude that 
\begin{equation}
    F_C'(x_1) = F_G'(x_2),\quad \exists x_1\in(0,x_3),\exists x_2\in(x_3, \infty),
\end{equation}
which contradicts the property in Equation \eqref{equ: r2_sufficiency_1a}. Therefore, our initial assumption must be false, thus the sufficiency is proven.
\end{proof}

\begin{proof}[Proof of necessity]
We may first assume the opposite of the conclusion (i.e., the Cauchy CDF is an overbound for the Gaussian CDF), then 
\begin{equation}
    F_G(x)>F_C(x),\quad \forall x>0. \label{equ: r2_tangential1}
\end{equation}
When $x>0$, both the Cauchy CDF and Gaussian CDF are strictly concave because their second-order derivatives are always negative. As shown in Fig. \ref{fig: tangential link}b, passing through any point A$(x_1, F_G(x_1))$ ($\forall x_1>0$) on Gaussian CDF, there exists a unique tangential line $l_1$ for $F_G(x)$, and the property of concavity gives that
\begin{equation}
    l_1 \geq F_G(x),\quad \forall x>0. \label{equ: r2_tangential2}
\end{equation}
Combining Equation  \eqref{equ: r2_tangential1} and \eqref{equ: r2_tangential2} , we obtain
\begin{equation}
    l_1 >F_C(x),\quad \forall x>0,
\end{equation}
which means line $l_1$ does not have any intersection with Cauchy CDF $F_C(x)$ when $x>0$. This further implies that the line $l_1$ can never be a tangential line for $F_C(x)$, and that the tangential link between Cauchy CDF and Gaussian CDF does not exist, which contradicts with the hypothesis (i.e., tangential link exists). Therefore, our initial assumption must be false, and the necessity is proven.
\end{proof}\par
In conclusion, the tangential link can always exist between Cauchy and Gaussian overbounds, if and only if the Cauchy CDF is not an overbound for the Gaussian CDF. Notably, the latter condition rarely happens unless the empirical error curve is approximately a Gaussian distribution. In that case, (take Figure \ref{fig: tangential link}a for instance), points O and B will be considered coincident, and the proposed Cauchy-Gaussian overbound will become a single-CDF Gaussian overbound. Since the focus of this study is to develop an overbound for heavy-tailed error distributions, we may neglect the extreme case and conclude that the tangential link between CDFs is always mathematically possible.
}

\section{Proof of overbound in the tangential transition region}
\label{app 2: ovb during transition in s.u. cases}
In the following proof, we only consider the case when $x>0$, because the conclusion can be similarly expanded to the domain of $x\leq 0$. 

Assume the CDF of the original error profile is $F_e(x)$ and the general form of tangential transition $T(x)$ in CDF from Cauchy to Gaussian gives:
\begin{equation}
    T(x) = k_m x+b \quad 0<x_1 \leq x \leq x_2,
\end{equation}
where $k_m$ and $b$ mathematically denote the slope and vertical intercept of the line function, respectively; $x_1$ and $x_2$ are the abscissa of the two tangential points (i.e.,  $(x_1, F_C(x_1))$ and $(x_2, F_G(x_2))$), when the transition happens from Cauchy to Gaussian model. According to the definition in Equation \eqref{equ: sym cdf x>0} and properties of the tangential line segment, we have:
\begin{align}
    T(x_1) = F_C(x_1)  &\leq F_e(x_1), \label{aequ: tangential x1}\\
   T(x_2)  = F_G(x_1) &\leq F_e(x_2), \label{aequ: tangential x2}
\end{align}
where $F_C(x)$ and $F_G(x)$ represent the CDFs of Cauchy and Gaussian, respectively.

For a s.u. error, the PDF $f_e(x)$ (i.e., $F_e'(x)$) is monotonically decreasing when $x>0$. For the convenience of proof, a function $m(x)$ is introduced as:
\begin{equation}
    m(x) = F_e(x) - T(x) \quad  0<x_1 \leq x \leq x_2.
\end{equation}
Combining with inequalities (\ref{aequ: tangential x1}) and (\ref{aequ: tangential x2}), we may obtain that:
\begin{align}
    m(x_1)&\geq 0 , \\
    m(x_2)&\geq 0.
\end{align}
The first derivative of $m(x)$ can be derived as:
\begin{equation}
    m'(x) = F_e'(x) - k_m \quad  x_1 \leq x \leq x_2.
\end{equation}

If $k_m \leq x_1 \leq x_2$, there will be $m(x)\leq 0$, thus $m(x)$ is decreasing when $x\in[x_1, x_2]$. The minimum $m(x)_{min}=m(x_2)\geq 0$, which further indicates $F_e(x) \geq T(x)\ \ \forall x\in[x_1, x_2]$.

If $x_1 \leq x_2 \leq k_m$, there will be $m(x)\geq 0$, thus $m(x)$ is increasing when $x\in[x_1, x_2]$. The minimum $m(x)_{min}=m(x_1)\geq 0$, which gives $F_e(x) \geq T(x)\ \ \forall x\in[x_1, x_2]$.

If $x_1< k_m < x_2 $, when $x\in[x_1, k_m]$,  $m(x)\geq 0$ thus $m(x)$ is increasing; when $x\in[k_m, x_2]$,  $m(x)\leq 0$ thus $m(x)$ is decreasing.  It is easy to determine the minimum $m(x)_{min}= \min(m(x_1), m(x_2))\geq 0$, which gives $F_e(x) \geq T(x)\ \ \forall x\in[x_1, x_2]$.

Combining all three scenarios, we can prove that curve $F_e(x)$ is always located above the transition $T(x)$ when $x>0$. Due to the symmetry of CDFs, a similar conclusion can be extended that $F_e(x)$ always lies below the transition $T(x)$ when $x\leq 0$. As such, it indicates that the overbounding properties are guaranteed during the two transition regions (i.e. when $0<x_1<x<x_2$ and $-x_2<x<-x_1\leq0$).

{
\section{Proof of Equation (33)}
\label{app 3: proof of f_VPE 1}
According to Equation \eqref{equ: VPE}, the $i^{th}$ component of VPE can be expressed by
\begin{equation}
    \text{VPE}_i =\mathbf{S}_{3,i}\varepsilon_i,
\end{equation}
then the CDF ($F_{\text{VPE}_i}$) can be derived by
\begin{align}
    F_{\text{VPE}_i}(x)&=\Pr\left(\text{VPE}_i<x\right) = \Pr\left(\varepsilon_i<\frac{x}{|\mathbf{S}_{3,i}|}\right) = \int_{-\infty}^{\frac{x}{|\mathbf{S}_{3,i}|}} f_{\varepsilon_i}(t)dt.
\end{align}
Let $u=|\mathbf{S}_{3,i}|t$, we have $du=|\mathbf{S}_{3,i}|dt$, and further obtain the CDF and PDF ($f_{\text{VPE}_i}$) as:
\begin{align}
    F_{\text{VPE}_i}(x) &=\frac{1}{|\mathbf{S}_{3,i}|}\int_{-\infty}^{x}f_{\varepsilon_i}\left(\frac{u}{|\mathbf{S}_{3,i}|}\right)du, \\
    f_{\text{VPE}_i}(x) & =\frac{1}{|\mathbf{S}_{3,i}|}f_{\varepsilon_i}\left(\frac{x}{|\mathbf{S}_{3,i}|}\right).
\end{align}
The distribution of VPE ($f_{\text{VPE}}$) is the joint distribution of all $f_{\text{VPE}_i}$, which gives  
\begin{equation}
    f_{\text{VPE}}(x)= \Conv_{i=1}^{N}  f_{\text{VPE}_i}(x) =\left( \prod_{i=1}^{N} \frac{1}{|S_{3,i}|}\right)   \left(\Conv_{i=1}^{N}  f_{\bm{\varepsilon}_i}\left(\frac{x}{|S_{3,i}|} \right) \right).
\end{equation}
}

{
\section{NavDEN model}
\label{app 4: navden model}
The Navigation Discrete ENvelope (NavDEN) is a discrete error model designed to provide a conservative bound on GNSS heavy-tailed measurement errors. It features a Gaussian-like core to tightly bound nominal errors and flared, non-Gaussian tails to conservatively account for the higher-than-Gaussian probability of large errors. The model is defined by a set of centrosymmetric left and right bounds on a regularly spaced grid.

% Multiple parameters in NavDEN overbound control the grid spacing ($\Delta$); represent the growth rate and asymptote at the tail regions ($\tilde{x}_{max}$, $\tilde{x}_{min}$, $\tilde{B}$, $\tilde{C}$); categorize the error coordinate into left-tail, core, and right-tail regions ($k_{tr}$, $k_{max}$, $k_{min}$); and determine the lateral offset for the paired bounds ($k_{bias}$).

The NavDEN overbound includes three regions: left tail ($K_1$), core ($K_2$), and right tail ($K_3$). Each region contains a subset of envelope boundaries noted by integer indices as below. 
\begin{align}
    K_1 &= [k_{min}, -k_{tr}), \\
    K_2 &= [-k_{tr}, k_{tr}], \\
    K_3 &= (k_{tr}, k_{max}],
\end{align}
where $k_{max}$ and $k_{min}$ are the largest and smallest index, respectively; $k_{tr}$ performs as the transition index between the core and tail regions. The normalized left bound, $\tilde{l}_k$, for an envelope with integer index $k$ is defined by the following piecewise function: 
\begin{equation}
    \tilde{l}_k= \left\{
        \begin{array}{lr} 
          \text{floor}\left(\tilde{C}\ln\left(1-\frac{k+k_{tr}}{k_{min}+k_{tr}}\right)-k_{tr}-k_{bias} \right)        &  k\in K_1 \\
         k-k_{bias}                             &  k\in K_2 \\
          \text{floor}\left(\tilde{x}_{max}-k_{bias}-\left(\tilde{x}_{max}-k_{tr}\right)e^{2(k_{tr}-k)/\tilde{B}}\right)        &  k\in K_3
        \end{array},
    \right.
\end{equation}
where parameters with a tilde notation are normalized by the fundamental grid spacing, $\Delta$. The model is defined by several parameters, including the transition index ($k _{tr}$), asymptotes ( 
$x_{min}$, $x_{max}$), curvature parameters ($ \tilde{B}$, $\tilde{C}$ ), and a core offset ($k_{bias}$). The cumulative probability associated with each left bound, expressed in Gaussian-quantile form ($\tilde{G}_{l,k}$), is given by :
\begin{equation}
    \tilde{G}_{l,k}= \left\{
        \begin{array}{lr} 
          -k_{tr}+\psi_1(k+k_{tr})        &  k\in K_1 \\
         k                             &  k\in K_2 \\
          k_{tr}+\psi_2(k-k_{tr})        &  k\in K_3
        \end{array},
    \right.
\end{equation}
where $\psi_1=\frac{\tilde{x}_{min}+k_{tr}}{k_{min}+k_{tr}}$ and $\psi_2=\frac{\tilde{x}_{max}-k_{tr}}{k_{max}-k_{tr}}$. Due to the model's symmetry, the right bounds ($\tilde{r}_k$) and their associated Gaussian-quantiles ($G_{r,k}$) are simply a reflection of the left bounds around the origin. For a detailed explanation on NavDEN model, one may refer to the original work in \citep{rife_overbounding_2012}.
}

\section*{Acknowledgments}
This research was supported by the project H-ZGD2 \enquote{Huawei-PolyU High-Precision Positioning with Vehicle-Mounted Multi-Sensor Fusion in Complex Scenarios (Phase II)}, the project 1-BBDW \enquote{Multi-robot Collaborative Operations in Lunar Areas for Regolith Processing}, and the Innovation and Technology Fund - Innovation and Technology Support Programme (ITF-ITSP) of the Innovation and Technology Commission, under Grant No. ITS/091/23, \enquote{Improving the Smartphone Urban Positioning Accuracy aided by AI and Probabilistic Graphical Model based Multi-sensory Integration for Hong Kong Smart Cities.}

\section*{Conflict of Interest}
The authors declare no conflicts of interest.

\printbibliography

\end{document}